\documentclass[11pt]{article}
\usepackage{amsmath,amsfonts,amssymb,amsthm}
\usepackage{fullpage}
\usepackage[colorlinks]{hyperref}
\hypersetup{linkcolor=cyan,filecolor=cyan,citecolor=cyan,urlcolor=cyan}
\usepackage{xspace}
\usepackage{thm-restate,color,xcolor}
\usepackage{boxedminipage}
\usepackage[boxed]{algorithm}
\usepackage{epigraph}
\usepackage[sc]{mathpazo}
\usepackage{amsmath}
\usepackage{mathtools}
\usepackage{framed}
\usepackage[framemethod=tikz]{mdframed}
\usepackage{titlesec}
\usepackage{lipsum}%
\usepackage{cleveref,aliascnt}
\usepackage{tikz}
\usepackage{wrapfig}
\usepackage{framed}
\usepackage[framemethod=tikz]{mdframed} 

\usetikzlibrary{shapes.geometric}
\usetikzlibrary{arrows}
\usetikzlibrary{arrows.meta}
\usetikzlibrary{patterns}
\usetikzlibrary{shapes.misc}

\newcommand{\dist}{\mbox{\rm dist}}
\newcommand{\len}{\mbox{\rm len}}
\newcommand{\inn}{in}
\newcommand{\outt}{out}

\newcommand{\ShortcutSmallD}{\mathsf{ShortcutSmallDiam}}
\newcommand{\ShortcutLargeD}{\mathsf{ShortcutLargeD}}
\newcommand{\HopsetLargeHop}{\mathsf{HopsetLargeHop}}
\newcommand{\HopsetSmallHop}{\mathsf{HopsetSmallHop}}
\newcommand{\DynamicAPSP}{\mathsf{DynamicAPSP}}

\theoremstyle{plain}
\newtheorem{thm}{Theorem}[section]
\newcommand{\BTHM}{\begin{thm}} \newcommand{\ETHM}{\end{thm}}
\newtheorem{cor}[thm]{Corollary}
\newcommand{\BCR}{\begin{cor}} \newcommand{\ECR}{\end{cor}}
\newtheorem{lem}[thm]{Lemma}
\newcommand{\BL}{\begin{lem}}   \newcommand{\EL}{\end{lem}}
\newtheorem{clm}[thm]{Claim}
\newcommand{\BCM}{\begin{clm}}   \newcommand{\ECM}{\end{clm}}
\newtheorem{prop}[thm]{Proposition}
\newcommand{\BP}{\begin{prop}}   \newcommand{\EP}{\end{prop}}
\newtheorem{assm}[thm]{Assumption}
\newcommand{\BASM}{\begin{assm}}   \newcommand{\EASM}{\end{assm}}
\newtheorem{question}[thm]{Question}

\theoremstyle{definition}
\newtheorem{defn}{Definition}[thm]
\newcommand{\BD}{\begin{defn}}   \newcommand{\ED}{\end{defn}}
\newtheorem{con}[thm]{Conjecture}
\newcommand{\BCONJ}{\begin{con}}   \newcommand{\ECONJ}{\end{con}}

\theoremstyle{definition}
\newtheorem{problem}[thm]{Problem}
\newcommand{\BPR}{\begin{problem}}   \newcommand{\EPR}{\end{problem}}

\newenvironment{rem}{\noindent{\bf Remark:~~}}{}
\newcommand{\BREM}{\begin{rem}} \newcommand{\EREM}{\end{rem}}
\newenvironment{discussion}{\noindent{\bf Discussion:~~\\}}{}
\newcommand{\BDIS}{\begin{discussion}} \newcommand{\EDIS}{\end{discussion}}
\newtheorem{obs}{Observation}[section]
\numberwithin{equation}{section}

\def\blackslug
     {\hbox{\hskip 1pt\vrule width 8pt height 8pt depth 1.5pt\hskip 1pt}}
\def\qed{\quad\blackslug\lower 8.5pt\null\par}

\newcommand{\poly}{{\rm poly}}

\newcommand{\Diam}{D}

\newtheorem{exmp}[thm]{Example}
\newtheorem{fact}[thm]{Fact}
\newcommand{\BEX}{\begin{exmp}} \newcommand{\EEX}{\end{exmp}}
\newcommand{\BF}{\begin{fact}}   \newcommand{\EF}{\end{fact}}
\newcommand{\Bcr}{\begin{techcorr}}
\newcommand{\Ecr}{\end{techcorr}}
\newcommand{\BDS}{\begin{description}}
\newcommand{\EDS}{\end{description}}
\newcommand{\BE}{\begin{enumerate}}
\newcommand{\EE}{\end{enumerate}}
\newcommand{\BI}{\begin{itemize}}
\newcommand{\EI}{\end{itemize}}
\newcommand{\BPF}{\begin{proof}}
\newcommand{\EPF}{\end{proof}}

\newcommand{\BB}{\begin{enumerate}}
\newcommand{\EB}{\end{enumerate}}

\usepackage{geometry}
\title{New Diameter-Reducing Shortcuts and Directed Hopsets: \\ Breaking the $O(\sqrt{n})$ Barrier}

\author{
Shimon Kogan \\
        \small Weizmann Institute\\
        \small shimon.kogan@weizmann.ac.il
\and				
Merav Parter \thanks{This project is funded by the European Research Council (ERC) under the European Union’s Horizon 2020 research and innovation programme (grant agreement No. 949083).}\\
        \small Weizmann Institute \\
        \small merav.parter@weizmann.ac.il
}
\date{}

\begin{document}
\maketitle

\begin{abstract}
For an $n$-vertex digraph $G=(V,E)$, a \emph{shortcut set} is a (small) subset of edges $H$ taken from the transitive closure of $G$ that, when added to $G$ guarantees that the diameter of $G \cup H$ is small. Shortcut sets, introduced by Thorup in 1993, have a wide range of applications in algorithm design, especially in the context of parallel, distributed and dynamic computation on directed graphs.  A folklore result in this context shows that every $n$-vertex digraph admits a shortcut set of linear size (i.e., of $O(n)$ edges) that reduces the diameter to\footnote{The notation $\widetilde{O}(\cdot)$ hides poly-logarithmic terms in $n$.} $\widetilde{O}(\sqrt{n})$. Despite extensive research over the years, the question of whether one can reduce the diameter to $o(\sqrt{n})$ with $\widetilde{O}(n)$ shortcut edges has been left open. 

We provide the first improved diameter-sparsity tradeoff for this problem, breaking the $\sqrt{n}$ diameter barrier. Specifically, we show an $O(n^{\omega})$-time randomized algorithm\footnote{Where $\omega$ is the optimal matrix multiplication constant.} for computing a linear shortcut set that reduces the diameter of the digraph to $\widetilde{O}(n^{1/3})$. This narrows the gap w.r.t the current diameter lower bound of $\Omega(n^{1/6})$ by [Huang and Pettie, SWAT'18]. Moreover, we show that a diameter of $\widetilde{O}(n^{1/2})$ can in fact be achieved with a \emph{sublinear} number of $O(n^{3/4})$ shortcut edges. Formally, letting $S(n,D)$ be the bound on the size of the shortcut set required in order to reduce the diameter of any $n$-vertex digraph to at most $D$, our algorithms yield:
\[
S(n,D)=\begin{cases}
\widetilde{O}(n^2/D^3), & \text{for~} D\leq n^{1/3},\\
\widetilde{O}((n/D)^{3/2}), & \text{for~} D> n^{1/3}~.
\end{cases}
\]
We also extend our algorithms to provide improved $(\beta,\epsilon)$ hopsets for $n$-vertex weighted directed graphs. 
%
%
%
\end{abstract}

\newpage

\tableofcontents

\newpage
\setcounter{page}{1}
\section{Introduction}
Reachability and shortest path computations in directed graphs are among the most fundamental graph theoretic problems. In many computational settings (e.g., parallel and distributed), the key complexity measures for solving these problems depend on the graph's diameter. The diameter of a directed graph $G$ is the minimum value $d$ such that any $u$-$v$ shortest path in $G$, if exists, has at most $d$ edges. For example, in the PRAM model, the span of many graph problems depends on the diameter of the graph, e.g., \cite{LiuJS19}.
In the context of data-structures, the tradeoff between the size of the graph and its diameter determines the space vs. query time (respectively) of reachability queries \cite{Hesse03}. 

\indent Motivated by this inherent dependency on the graph's diameter, Thorup \cite{Thorup92} introduced the notion of \emph{diameter shortcutting sets}. Informally, this refers to a set of edges taken from the transitive closure of the graph $G$ whose addition to $G$ reduces the diameter of the graph substantially. In other words, the shortcut edges preserve the transitive closure of the graph (and possibly also the distances) without increasing much the size of the graph. The key complexity measures are the number of shortcut edges and the resulting diameter of the augmented graph. The na\"{\i}ve solution of adding all edges of the transitive closure provides a perfect diameter bound of $1$, but at the cost of adding $\Theta(n^2)$ edges, which is clearly undesirable. 

Since their introduction, shortcutting sets have been studied quite extensively from a combinatorial and algorithmic perspectives. They currently provide the core computational step in the parallel, distributed and dynamic computation of reachability and directed shortest paths \cite{KleinS97,HenzingerKN14,HenzingerKN15,ForsterN18,LiuJS19,Fineman20,GutenbergW20a,BernsteinGW20}. 
To this date, the only known tradeoff for these structures is given by a folklore randomized algorithm, attributed to Ullman and Yannakakis \cite{UllmanY91}, that for every integer $D\geq 1$, computes a shortcut set of size $O((n/D)^2\log^2 n)$ and provides a diameter bound of $D$. Berman et al. \cite{BermanRR10} improved the size bound to $O((n/D)^2)$. Very recently, there has been a sequence of breakthrough results that compute shortcut edges that (almost) achieve this tradeoff in almost optimal time in the sequential, parallel and distributed settings \cite{Fineman20,LiuJS19,KarczmarzS21}. Liu, Jambulapati and Sidford \cite{LiuJS19} extended the framework of Fineman \cite{Fineman20} to compute, in nearly linear time, a shortcut set of $\widetilde{O}(n)$ edges that reduces the diameter of the graph to $D=O(n^{1/2 + o(1)})$ (i.e., almost as obtained by the folklore algorithm). In their paper, \cite{LiuJS19} noted that:
\begin{quote}
\emph{From the perspective of constructability the true tradeoff between the diameter and number of added edges is not known..it is not known how to improve the straightforward random construction in any regime.}
\end{quote}
%
%
%
%
%
%
%
The question of whether one can improve this tradeoff has been stated explicitly and implicitly in many of the prior works on this topic. Since most algorithmic applications mainly focus on linear-size shortcuts, we ask:
\begin{question}
Is it possible to break the $\sqrt{n}$ diameter barrier by adding $\widetilde{O}(n)$ shortcut edges?
\end{question}
While not much progress has been provided on the upper bound side, there has been more movement on the lower bound aspects of the problem. Hesse \cite{Hesse03} presented a construction of an $n$-vertex digraph with $m = \Theta(n^{19/17})$ edges and diameter $d=\Theta(n^{1/17})$, that requires $\Omega(m n^{1/17})$ shortcut edges to reduce its diameter to $o(d)$. This lower bound result refutes Thorup's conjecture \cite{Thorup92} concerning the existence of linear-size shortcuts that reduce the diameter of any graph to poly-logarithmic\footnote{This conjecture in fact holds for special graph families, such as planar graph \cite{Thorup92}.}.
Huang and Pettie \cite{HuangP18} extended and improved the lower bound of Hesse \cite{Hesse03} by presenting a construction of $n$-vertex $m$-edge graph for which any $O(m)$-size shortcut set cannot reduce the diameter to below $\Omega(n^{1/11})$, and any $O(n)$-size shortcut set cannot reduce the diameter to below $\Omega(n^{1/6})$. Hesse \cite{Hesse03} concludes his lower-bound paper by saying:
\begin{quote}\cite{Hesse03}
\emph{Achieving tight bounds on the problem of reducing diameter by shortcutting directed graphs will undoubtedly require ingenious new constructions or algorithms, which will be a valuable contribution to theoretical computer science.}
\end{quote}  

In this paper, we provide the first combinatorial progress over the folklore algorithm. Our constructions improve the tradeoff in the entire regime. For example, we present a randomized algorithm that computes a collection of $\widetilde{O}(n)$ shortcut edges and obtains a diameter bound of $O(n^{1/3})$. We also show that a diameter of $O(\sqrt{n})$ can in fact be achieved by adding only $\widetilde{O}(n^{3/4})$ edges. The precise formulation of our new tradeoffs appear in Theorem \ref{main_theorem_math_ver1}. Our algorithm is based on a new shortcutting framework, that in comparison to the folklore algorithm, randomly samples both vertices and \emph{paths} from a given collection of precomputed path collection. We hope that our techniques will pave the way towards improved parallel and distributed algorithms in digraphs. 

\paragraph{Approximate Directed Hopsets.} In the context of distance computation in directed graphs, it is desired to reduce the number of hops of directed shortest paths. The \emph{shortest path diameter} of a (possibly weighted) digraph $G$ is the minimum value $d$ such that every pair $u,v$ has a shortest path in $G$ with at most $d$ edges (denoted as \emph{hops}). Similarly to reachability, the parallel and distributed computation of directed shortest paths also heavily depend on the shortest path diameter of the graph. For any $\epsilon \geq 0$, and a positive integer $\beta$, an $(\beta,\epsilon)$-hopset is a set of weighted shortcut edges which, when added to the graph, admit $\beta$-hop paths of weight at most $(1+\epsilon)$ time the true shortest path distances. Formally, $(\beta,\epsilon)$-hopsets are defined as follows. For any vertices $u,v \in V$, define $\dist_G^{(\beta)}(u,v)$ to be the minimum length $u$-$v$ path with at most $\beta$ edges (hops). If there is no such path, then $\dist_G^{(\beta)}(u,v)=\infty$. 

\begin{defn}[$(\beta,\epsilon)$\textbf{-hopsets}] An $(\beta,\epsilon)$-hopset for a weighted digraph $G=(V,E,\omega)$ is a set of weighted edges $H$, such that, for any vertex pair $u,v \in V$, it holds that
$$\dist_{G}(u,v) \leq \dist_{G \cup H}^{(\beta)}(u,v)\leq (1+\epsilon)\dist_{G}(u,v)~.$$
The parameter $\beta$ is the \emph{hopbound}, and $|H|$ is the size of the hopset.
\end{defn}

For $\epsilon=0$, the $(\beta,\epsilon)$ hopset is denoted as \emph{exact hopset}. Hopsets also suffer from the $\sqrt{n}$ barrier on the number of hops achievable with a linear shortcut set. Cao, Fineman and Russell \cite{CaoFR20,CaoFRSPAA20,CaoFR21} presented the first nearly work-efficient parallel algorithms with $\widetilde{O}(m)$ work and span of $n^{1/2+o(1)}$. At the heart of their solution lies a computation of $(\beta,\epsilon)$ hopsets with hopbound $\beta=n^{1/2+o(1)}$ with $\widetilde{O}(n)$ edges (i.e., almost the same tradeoff as that of the folklore algorithm). 

In contrast to directed graphs, for undirected graphs hopsets with improved bounds are known \cite{Cohen00,ShiS99,KleinS97}. Elkin and Nieman \cite{ElkinN19} and Huang and Pettie \cite{HuangP19} showed that Thorup-Zwick emulators \cite{ThorupZ06} provide nearly optimal $(\beta,\epsilon)$ hopsets with hopbound of $\beta=(\log k/\epsilon)^{\log k}$ and $O(n^{1+1/k}\log n)$ edges, for every integer $k\geq 1$. Abboud, Bodwin and Pettie \cite{AbboudBP18} proved that these constructions are nearly tight for constant values of $k$. We note that the problem of computing tight $(\beta,\epsilon)$ hopsets for \emph{directed} graphs, already in the unweighted case, is fairly open. Due to the extra requirement to approximate also the distances, computing $(\beta,\epsilon)$ hopsets is in general more involved than diameter shortcutting sets. 

\subsection{Our Contribution.}\label{sec:cont}
For every positive integers $n\geq D$, let $S(n,\Diam)$ denote the universal upper bound on the number of shortcut edges (i.e., edges from the transitive closure) required to be added to $n$-vertex digraphs in order to reduce its diameter to at most $D$. Let $G^{\star}$ denote the transitive closure of a digraph $G$. Our key contribution is in providing new shortcutting algorithms with improved
$S(n,\Diam)$ bounds that in particular, break the long-standing $\sqrt{n}$ barrier. 
\BTHM\label{main_theorem_math_ver1}[New Diameter Shortcutting Sets]
For every graph $G$ on $n$ vertices, there is an $\tilde{O}(n^\omega)$-time randomized algorithm for computing a shortcut set of size:
\[
S(n,D) = 
\begin{cases}
\tilde{O}\left( n^2/D^3\right), &\text{for $D \leq n^{1/3}$;} \\
\tilde{O}( (n/D)^{3/2}),  &\text{for $D > n^{1/3}$,}
\end{cases}
\]
where $\omega$ is the optimal matrix multiplication constant. 
\ETHM
This improves considerably over the state-of-the-art bound of $S(n,D)=O((n/D)^{2})$ obtained by the folklore algorithm \cite{UllmanY91,BermanRR10}. These bounds should be compared with the 
lower bound result of Huang and Pettie \cite{HuangP18} that provides $S(n,\Diam)=\Omega(n)$ for $\Diam=O(n^{1/6})$. 
In addition, for every $0 <\epsilon <1/4$ and $\delta>0$, Hesse \cite{Hesse03} provided a diameter lower bound of $\Omega(n^{\delta})$ obtained with $O(n^{2-\epsilon})$ edges, i.e., $S(n,\Diam)=\Omega(n^{2-\epsilon})$ for $D=O(n^{\delta})$. 
The literature has also considered $O(m)$-size shortcuts for $n$-vertex graphs with $m$ edges, referred to as \emph{matching shortcuts} by \cite{Thorup92}. Solving the bounds of Theorem \ref{main_theorem_math_ver1}
for $n^2/D^3=m$, provides a matching shortcut with diameter 
$D=\widetilde{O}\left(\left(\frac{n^2}{m} \right)^{1/3}\right)$.
%
This improves over the diameter bound of $\widetilde{O}(n/\sqrt{m})$ obtained by the folklore algorithm.
On the lower bound side, Huang and Pettie \cite{HuangP18} presented a construction of $n$-vertex graph for which any matching shortcuts cannot reduce the diameter to below $\Omega(n^{1/11})$. We next generalize the lower bound constructions of \cite{Hesse03,HuangP18} to provide a smooth tradeoff between the diameter and the number shortcut sets. 
\BTHM\label{pettie_lower_bound}[Extended Lower Bound Constructions]
For every $n$ and $D \geq n^{1/6}$, there is an $n$-vertex digraph $G$ for which any addition of $\Theta((n/D)^{6/5})$ shortcut edges in $G^{\star}$ cannot reduce the diameter to below $D$. Thus $S(n,D)=\Omega((n/D)^{6/5})$. 
\ETHM
In particular, for $D=n^{1/3}$, we get $S(n,D) = \Omega(n^{4/5})$ and $S(n,D) = \widetilde{O}(n)$. For $D=\sqrt{n}$, we have $S(n,D) = \Omega(n^{3/5})$ and $S(n,D) = \widetilde{O}(n^{3/4})$. Therefore, the new upper and lower bound results of Theorems \ref{main_theorem_math_ver1} and \ref{pettie_lower_bound} considerably tighten the bounds on the $S(n,D)$ function. 

\paragraph{Implication for Transitive-Closure (TC) Spanners.} For a given (unweighted) digraph $G=(V,E)$ with transitive closure $G^{\star}$, a subgraph $H \subseteq G^{\star}$ is a $k$-TC spanner of $G$ if 
for every $(u,v)\in G^{\star}$ it holds that $\dist_H(u, v) \leq k$. TC-spanners have been formally defined by Bhattacharyya et al. \cite{BhattacharyyaGJRW09}, and have shown to be useful in many different contexts, such as access control, property testing, and data structures. The interested reader is encouraged to read the beautiful survey on TC spanners by Raskhodnikova \cite{Raskhodnikova10}. The size vs. stretch tradeoff of TC spanners is very closely related to that of the diameter shortcutting sets. TC spanners, however, have been studied mainly from an \emph{approximation} perspective. As a direct corollary of Theorem \ref{main_theorem_math_ver1}, we also improve the state-of-the-art approximation ratio for this problem.  In particular, by employing the transitive reduction technique of Aho, Garey and Ullman \cite{AhoGU72}, we have:
\begin{cor}\label{cor:TC}
There is an algorithm that provides an approximation ratio of $\widetilde{O}(n/k^3)$ for $k$-TC spanners. 
\end{cor}
This improves upon the 10-year old state-of-the-art approximation ratio of $O(n/k^2)$ by Bhattacharyya et al. \cite{BhattacharyyaGJRW09} and Berman \cite{BermanRR10}.

\paragraph{Approximate Directed Hopsets.} We thrn turn to consider new constructions of $(\epsilon,\beta)$ hopsets for weighted $n$-vertex digraphs with integer weights in $[1,W]$.  Let $H(n,W,\beta,\epsilon)$ be the universal upper bound on the number of edges required to be added to $n$-vertex digraphs with weights $[1,W]$ in order to provide an $(1+\epsilon)$ approximate $\beta$-hop distances. We show:

\BTHM\label{existential_hopset_thm1}[New $(\beta,\epsilon)$ Directed Hopsets]
For every $n$-vertex graph $G$ with integer weights $[1,W]$, $\epsilon \in (0,1)$ and $\beta \in \mathbb{N}_{\geq 1}$, there is an $\widetilde{O}(n^3 \cdot \poly\log(n W))$-time randomized algorithm for computing  $(\epsilon,\beta)$ hopsets whose number of edges is bounded by:
\[
H(n,W,\beta,\epsilon) = 
\begin{cases}
\widetilde{O}\left( \left(\frac{n^2}{\beta^3}\right)  \epsilon^{-2} \log (nW)\right),  &\text{for $\beta \leq n^{1/4}$;} \\
\widetilde{O}\left( \left(\frac{n}{\beta}\right)^{5/3} \epsilon^{-2} \log (n W) \right), &\text{for $\beta > n^{1/4}$}
\end{cases}
\]
\ETHM
Thus, for example, for $W=\poly(n)$, one can provide an $(\beta,\epsilon)$ hopsets with $\widetilde{O}(n)$ edges, and hopbound of $\beta=n^{2/5}$, improving upon the state-of-the-art hopbound of $\sqrt{n}$. 
Technically, these constructions are more involved than that of diameter shortcutting sets for the following reason. For the latter case, one can assume, w.l.o.g, that the input graph is directed acyclic, as it is sufficient to shortcut each strongly connected component, and then to consider the graph obtained by contracting each such component. This assumption is no longer valid when considering distances, and hence $(\beta,\epsilon)$ hopsets require new ideas as elaborated in Sec. \ref{sec:shortcut}. 

\paragraph{Open Problems.} Our work leave several intriguing open problems. The first natural objective is to further narrow the gap between the upper and lower diameter bounds. For example, for linear number of shortcut edge, the diameter upper bound is $\widetilde{O}(n^{1/3})$, and the lower bound is $\Omega(n^{1/6})$ by \cite{HuangP19}. Another intriguing setting is concerned with \emph{exact} directed hopsets, for which the $\sqrt{n}$ barrier still holds. This raises now the question (already for unweighted graphs) of whether there is a real separation between approximate and exact hopsets. 
\vspace{-7pt}\subsection{A New Shortcutting Approach.}\label{sec:shortcut}
For a given $n$-vertex digraph $G$, let $G^{\star}$ denote its transitive closure. For a vertex subset $V'$, let $G^{\star}[V']=\{(u,v) \in G^{\star} ~\mid u,v \in V'\}$ be the induced transitive closure on $V'$. To provide a clean intuition into our algorithms, we restrict attention to the computation of $\widetilde{O}(n)$ shortcut edges, that provide a diameter bound of $O(n^{1/3})$. We will neglect for now the time complexity aspects and focus on the existential argument.  We start by describing the folklore algorithm that provides a diameter of $O(\sqrt{n})$. This algorithm samples a subset $V' \subset V$ of $O(\sqrt{n}\log n)$ vertices (w.h.p.), by sampling each vertex $v$ into $V'$ independently with probability $\Theta(\log n/\sqrt{n})$. The shortcut set $H$ is given by $H=G^{\star}[V']$, and thus $|H|=O(n\log^2 n)$. 
The bound on the diameter in $G \cup H$ follows by a simple hitting set argument. Consider a $u$-$v$ shortest path $P \subseteq G$ of at least $3\sqrt{n}$ edges\footnote{If no such path exists, the claim holds immediately.}. Let $P',P''$ denote the $\sqrt{n}$-length prefix (resp., suffix) of $P$. W.h.p., it holds that $P' \cap V'\neq \emptyset$ and also $P'' \cap V'\neq \emptyset$. Let $u',v' \in V'$ be some representative vertices on $P', P''$, respectively. Since $(u',v')\in H$, there is an $u$-$v$ path $Q=P[u,u']\circ (u',v') \circ P[v',v]$ in $G \cup H$ with $\Theta(\sqrt{n})$ edges. 

Our approach is based on a random sampling of \emph{vertices} and \emph{paths}. Throughout, we assume that the graph $G$ is a directed acyclic graph (DAG). As noted in prior works, e.g., \cite{Raskhodnikova10}, this assumption can be made w.l.o.g. We start by computing a maximal collection of vertex-disjoint dipaths $\mathcal{P}=\{P_1,\ldots, P_\ell\}$ where each dipath $P_i \subseteq G^{\star}$ has length $n^{1/3}$, and thus $|\mathcal{P}|=O(n^{2/3})$. By the maximality of $\mathcal{P}$, the remaining graph $G^{\star} \setminus \bigcup_{P_i \in \mathcal{P}} V(P_i)$ contains no dipath of length $n^{1/3}$. Letting $G'=G^{\star}$, then such a collection can be obtained by iteratively computing an $n^{1/3}$-length dipath $P$ in $G' \subseteq G^{\star}$, adding $P$ to the collection $\mathcal{P}$, and deleting the vertices of $P$ from $G'$. This should be repeated as long as the diameter of $G'$ is at least $n^{1/3}$. (We provide a \emph{time efficient} procedure for computing a path collection with similar properties using the decomposition into chain and anti-chains by Grandoni et al. \cite{GrandoniILPU21}.)

The set of shortcut edges added to $H$ consists of two subsets $H_1$ and $H_2$. 
The first set $H_1$ is obtained by computing a shortcut set $H(P_i)$ for every path $P_i \in \mathcal{P}$, such that $H(P_i)\cup P_i$ has diameter at most $2$, and in addition $|H(P_i)|=O(|V(P_i)|\cdot \log |V(P_i)|)$. Such shortcuts can be computed by a known recursive procedure (e.g., see \cite{Raskhodnikova10}). Since the paths in $\mathcal{P}$ are vertex-disjoint, we have that $H_1=\bigcup_{P_i}H(P_i)$ consists of $O(n\log n)$ edges. The second set $H_2$ is obtained by connecting a random sample of vertices $V'$ to a random sample of paths $\mathcal{P}' \subset \mathcal{P}$. Specifically, $V'$ is obtained by sampling each vertex $v \in V$ into $V'$ independently with probability $\Theta(\log n/ n^{1/3})$. Therefore, w.h.p., $|V'|=O(n^{2/3}\log n)$. The set $\mathcal{P}'$ is obtained by sampling each path $P \in \mathcal{P}$ into $\mathcal{P}'$ independently with probability $\Theta(\log n/ n^{1/3})$, thus $|\mathcal{P}'|=O(n^{1/3}\log n)$. 

For each vertex $v \in V'$ and a sampled path $P_i \in \mathcal{P}'$, we add to $H_2$ a unique edge $e(v,P_i)$ that captures the $G$-reachability of $\{v\} \times V(P)$. Letting $P_i$ be an $x \leadsto y$ path, then $e(v,P_i)=(v,w)$ where $w$ is the \emph{first} vertex on $P$ (closest to $x$) satisfying that $(v,w)\in G^{\star}$. Overall, $|H_2|=O(|V'|\cdot |\mathcal{P}'|)=O(n\log^2 n)$ edges, w.h.p. We next bound the diameter of $G \cup H_1 \cup H_2$. 

Let $P$ be some $u$-$v$ \emph{shortest} path in $G \cup H_1$ of length at least $20n^{1/3}$. If no such path exists, the diameter bound holds. We have the following:
\begin{itemize}
\item{(Q1)} The path $P$ contains at most $n^{1/3}$ vertices in $V \setminus \bigcup_{P_i \in \mathcal{P}}V(P_i)$.
\item{(Q2)} $|V(P)\cap V(P_i)|\leq 3$ for every $P_i \in \mathcal{P}$.
\end{itemize}
Property (Q1) holds by the termination condition of the iterative process, and noting that any subset of $n^{1/3}$ vertices on a path $P$ induces a directed path of length $n^{1/3}$ in $G^{\star}$. 
Property (Q2) holds since $G$ is a DAG, $H_1$ contains the shortcut set $H(P_i)$ for every $P_i \in \mathcal{P}$, and by the fact that $P$ is a shortest path. Let $P'$ and $P''$ be the $10n^{1/3}$-length prefix (resp., suffix) of the path $P$. By the Chernoff bound, we have that, w.h.p., there is some $v^* \in V' \cap V(P')$. In addition, by properties (Q1) and (Q2), $P''$ contains representatives from at least\footnote{By (Q1), $P''$ contains at least $9n^{1/3}$ vertices that lie on paths in $\mathcal{P}$. By Property (Q2), there are at least $3n^{1/3}$ representatives from distinct paths.} $3n^{1/3}$ \emph{distinct} paths in $\mathcal{P}'$. By the Chernoff bound, we have that $P''$ contains a vertex $w$ that appears on some sampled path $P_i \in \mathcal{P}'$, w.h.p. Since $e(v^*,P_i) \in H_2$, we get a $u$-$v$ path $Q=P[u,v^*] \circ e(v^*, P_i) \circ P_i[x,w] \circ P[w,v]$ of length $O(n^{1/3})$ in $G \cup H_1 \cup H_2$, where $e(v^*,P_i)=(v^*,x)$. The diameter bound holds. 

\vspace{-7pt}\subsection{Preliminaries.} For an $n$-vertex digraph $G$, let $G^{\star}$ denote its transitive closure. For a vertex pair $u,v \in V(G)$ where $(u,v) \in G^{\star}$, let $\dist_G(u,v)$ be the length (measured by the number of edges) of the shortest dipath from $u$ to $v$. For $(u,v) \notin G^{\star}$, $\dist_G(u,v)=\infty$. For a subset $V' \subseteq V$, let $G[V']$ be the induced graph on $V'$. The graph diameter is denoted by $D(G)=\max_{(u,v) \in G^{\star}}\dist_{G}(u,v)$. We say that $u \leadsto_G v$ if there is a directed path from $u$ to $v$ in $G$, i.e., $(u,v) \in G^{\star}$. A \emph{shortcut} edge is an edge in $G^{\star}$.
A \emph{chain} (of size $k$) of $G$ is a subset of vertices $\{c_1, \ldots, c_k\}$ such that $[c_1,\ldots, c_k]$ is a directed path in G. An \emph{anti-chain} (of size $k$) of $G$ is
a subset of vertices $\{a_1,\ldots, a_k\}$ such that there is no edge between them in $G$. 

For an integer weighted digraph $G=(V,E,\omega)$ where $\omega: E \to [1,W]$ is the weight function. 
Let $\dist_G(u,v)$ at the weight of the shortest path from $u$ to $v$. 
Let $\len(P)$ be the \emph{length} of path $P$, measured by the sum of its weighted edges. For any vertices $u,v \in V$, define $\dist_G^{(\beta)}(u,v)$ to be the minimum length $u$-$v$ path with at most $\beta$ edges (hops). If there is no such path, then $\dist_G^{(\beta)}(u,v)=\infty$. 

%
%
We need the following lemma that computes for a directed path $P$ a shortcut set of size $O(|P| \log |P|)$ that provides a diameter of $2$. 
\BL\label{line_shortcut_lem1}[Lemma 1.1, \cite{Raskhodnikova10}]
Given a directed path $P$, one can compute in $O(|P| \log |P|)$ time, a collection of $O(|P| \log |P|)$ shortcut edges $H$ from the transitive closure of $P$ such that $\dist_{G \cup H}(u,v)\leq 2$ for every $u,v \in V(P)$. 
\EL
Given a digraph $G$, by $G^+$ we denote the result of contracting all strong components in $G$. Note that $G^+$ is acyclic directed graph (DAG) and can be computed in linear time, i.e., in $O(|E(G)|+|V(G)|)$ time. 

\noindent\textbf{Roadmap.} In Sec. \ref{sec:diam-shortcut}, we describe the constructions of new shortcutting sets, proving Theorem \ref{main_theorem_math_ver1} and Cor. \ref{cor:TC}. In Sec. \ref{sec:hopsets}, we present the algorithms for computing approximate hopsets, thus proving Theorem \ref{existential_hopset_thm1}. Finally, in Sec. \ref{pettie_appendix}, provide an extension of the lower bound construction of \cite{HuangP18} and prove Theorem \ref{pettie_lower_bound}.

\section{Improved Diameter Shortcutting Algorithms}\label{sec:diam-shortcut}

\subsection{Reducing to Diameter $D=O(n^{1/3})$.}
In this section, we provide an $O(n^{\omega})$-time randomized algorithm for computing a shortcut set of size $S(n,D)$ for every $D=O(n^{1/3})$, with high probability. 

\BTHM\label{thm:first_case_mainthm1}
For every $n$-vertex digraph $G(V,E)$ and a parameter $D\leq n^{1/3}$, one can compute, in time\footnote{$\omega$ is the optimal matrix multiplication constant.} $\widetilde{O}(n^\omega)$, an edge set $H \subseteq G^{\star}$ of cardinality $|H|=\widetilde{O}(n^2/D^3)$ such that for every $(u,v) \in G^{\star}$ it holds that $\dist_{H \cup G}(u,v) \leq D$ w.h.p. 
\ETHM
The next proposition allows us to assume that the input graph is a DAG.
\begin{prop}\label{DAG_reduction1}[Lemma 3.2, \cite{Raskhodnikova10}]
Let $H^+$ be a shortcut set for $G^+$. One can compute in $O(|E(G)|+|H|)$-time, a shortcut set $H$ for $G$ such that $|H|=O(|H^+|+n)$ and $\Diam(G \cup H)=O(\Diam(G^+ \cup H^+))$. 
\end{prop}
%
%
%
From that point on, we assume that $G$ is a DAG. We use the following decomposition by Grandoni et al. \cite{GrandoniILPU21} which partitions the vertices of the graph into directed paths and independent sets. 
%
\BD [Path and Independent Set Decomposition]
An $(a, b)$-decomposition of a directed acyclic graph (DAG) $G$ consists of a collection $\mathcal{P}=\{P_1,\ldots, P_a\}$ of $a$ chains in $G$, and a collection $\mathcal{Q}=\{Q_1,\ldots, Q_b\}$ of $b$ anti-chains of $G$ such that $\bigcup_{i=1}^a V(P_i) \cup \bigcup_{j=1}^b Q_j=V(G)$. 
\ED

\BTHM\label{thm:dag_decomposition1} [See \cite{GrandoniILPU21}, Theorem 3.2]
Let $G$ be an $n$-vertex DAG. For every $\ell \in [1,n]$, there is an $O(n^2)$-time deterministic algorithm for computing an $(\ell, 2n/\ell)$-decomposition of $G$.
\ETHM
For a given directed path $P=[u_1,\ldots, u_k]$ and a vertex $v$, let $u_i$ be the \emph{first} vertex on $P$ satisfying that $(v,u_i) \in G^{\star}$. The edge $(v,u_i)$ is denoted as \emph{the first incoming edge} from $v$ to $P$, represented by $e(v,P)=(v,u_i)$. Note that the augmented path $P \cup \{e(v,P)\}$ provides a directed path from $v$ to every vertex $u \in P$ such that $(v,u)\in G^{\star}$. For a given set $X$ and $p \in [0,1]$, let $X[p]$ be the subset obtained by sampling each element in $X$ independently with probability $p$. 

\begin{mdframed}[hidealllines=false,backgroundcolor=gray!00]
\center \textbf{Diameter Shortcutting Algorithm $\ShortcutSmallD$:}
\begin{flushleft}
\textbf{Input:} An $n$-vertex DAG $G$ and a diameter bound $D\leq n^{1/3}$. \\
\textbf{Output:} A shortcut set $H \subseteq G^{\star}$ satisfying that $\max\limits_{(u,v)\in G^{\star}}\dist_{G \cup H}(u,v)\leq D$.
\end{flushleft}

\vspace{-8pt}
\begin{enumerate}
\item Compute a $(\ell=(16n/D), 2n/\ell)$-decomposition $(\mathcal{P},\mathcal{Q})$ of $G^{\star}$ using Theorem \ref{thm:dag_decomposition1}.

\item For every $P_i \in \mathcal{P}$, let $H_i=H(P_i)$ be a shortcut set for reducing the diameter of $P_i$ to $2$ as obtained by Lemma \ref{line_shortcut_lem1}. Add $\bigcup_{P_i \in \mathcal{P}} P_i \cup H_i$ to $H$.  

\item Let $V'=V[p]$ and $\mathcal{P}'=\mathcal{P}[p]$ for $p=\Theta(\log n/D)$. 

\item For every $v \in V'$ and $P_i \in \mathcal{P}'$, add the edge $e(v,P_i)$ to the shortcut set $H$.
\end{enumerate}

\end{mdframed}
\noindent \textbf{Remark.} Observe that in contrast to the high-level description of Sec. \ref{sec:shortcut}, in our algorithm, we do not require the paths in $\mathcal{P}$ to be of length $n^{1/3}$. As will be shown in the analysis, the properties of the paths obtained by the $(\ell,2n/\ell)$ decomposition are sufficient for our purposes. 

\paragraph{Running Time.} We first show that the time complexity of the algorithm can be bounded by $\widetilde{O}(n^\omega)$. The computation of the transitive closure $G^{\star}$ can be done in $\widetilde{O}(n^\omega)$ time. 
The computation of the $(16n/D, D/8)$-decomposition by Theorem \ref{thm:dag_decomposition1} takes $O(n^2)$ time. 
The computation of the shortcut graphs for each path $P_i \in \mathcal{P}$ takes $O(|P_i|\log |P_i|)$ time, by Lemma \ref{line_shortcut_lem1}. Note that this holds as the paths in $\mathcal{P}$ are vertex-disjoint. 
Given a vertex $v$, a path $P$ and the transitive closure graph $G^{\star}$, the computation of the edge $e(v,P)$ can be done in $O(\log |P|)$ time via binary search. Therefore, computation the edges $e(v,P)$ for every $v, P \in V' \times \mathcal{P}'$ takes $O(|\mathcal{P}'| \cdot |V'|\cdot \log n)=\widetilde{O}(n^2/D^3)$.
Note that the running time is dominated by the computation of the transitive closure $G^{\star}$. 
\\ \\
\noindent \textbf{Edge Bound.} Since the paths in $\mathcal{P}$ are vertex-disjoint, $|\bigcup H(P_i)|=\widetilde{O}(n)$. In addition, $|\mathcal{P}|=O(n/D)$ and thus by Chernoff bound, w.h.p. it holds that $|\mathcal{P}'|=O(n\log n/ D^{2})$ and $|V'|=O(n\log n/D)$. Therefore, w.h.p., the last step adds $|V'|\cdot |\mathcal{P}'|= \widetilde{O}(n^2/D^3) $ edges. 
\\ \\
\noindent \textbf{Diameter bound.} 
Let $N=V(G)\setminus (\bigcup_{P_i \in \mathcal{P}}V(P_i))=\bigcup_{j} Q_j$. Therefore, $N$ is the union of $D/8$ independent sets in $\mathcal{Q}$. Let $H'=\bigcup_{P_i \in \mathcal{P}} H(P_i)$. For a fixed pair $u,v$, let $P_{u,v}$ be a shortest $u$-$v$ path in $G \cup H'$. 
\BL\label{lem:chromaticnum1}
The path $P_{u,v}$ contains at most $D/8$ vertices from $Q$.
\EL
\BPF
By Step (1) of the algorithm, each path in $G^{\star}[N]$ contains at most $D/8$ vertices. Therefore, since the vertices of any $G$-dipath induce a path in $G^{\star}$, we also have that any $G$-dipath in contains at most $D/8$ vertices of $N$. 
\EPF
\noindent Let $P', P''$ be the $(D/4)$-length prefix (resp., suffix) of $P_{u,v}$.
Observe that $P''$ contains at most $3$ representative vertices from each fixed path $P_i \in \mathcal{P}$ by Step (2) of the algorithm. This holds since for each $P_i \in \mathcal{P}$, the vertices appearing on $P_{u,v}$ must appear in the same order as in $P_i$, as $G$ is a DAG. That is, since there is $2$-hop path in $H_i$ between every two vertices in $P_i$, we have that $|V(P'')\cap V(P_i)|\leq 3$, for every $P_i \in \mathcal{P}$. 

Therefore by Lemma \ref{lem:chromaticnum1}, $P''$ contains representatives from $\Omega(D)$ paths in $\mathcal{P}$, and w.h.p., $P''$ contains at least one representative vertex $y \in P_i$ for some sampled path $P_i=[a_1,\ldots, a_k]\in \mathcal{P}'$. In addition, the prefix $P'$ contains at least one sampled vertex $x \in V'$ w.h.p. Therefore the edge $e(x,P_i)$ is in $H$. 
Let $e(x,P_i) = (x,z)$ for $z \in P_i[a_1,y]$.
Therefore the augmented graph $G \cup H$ contains a $u$-$v$ path $P'_{u,v}=P_{u,v}[u, x] \circ (x,z) \circ P_i[z,y] \circ P_{u,v}[y,v]$ where
\begin{eqnarray*}
\dist_{G \cup H}(u,v)&\leq & |P'_{u,v}|=|P_{u,v}[u, x]|+1+|P_i[z,y]|+|P_{u,v}[y,v]|
\\&\leq & D/4+1 +\dist_{P_i \cup H_i}(z,y)+D/4 \leq D/4+D/4+3\leq D~.
\end{eqnarray*}
See Fig. \ref{fig:small-d} for an illustration. 
This concludes the proof of Theorem \ref{thm:first_case_mainthm1}.

\begin{figure}[h!]
\begin{center}
\includegraphics[scale=0.35]{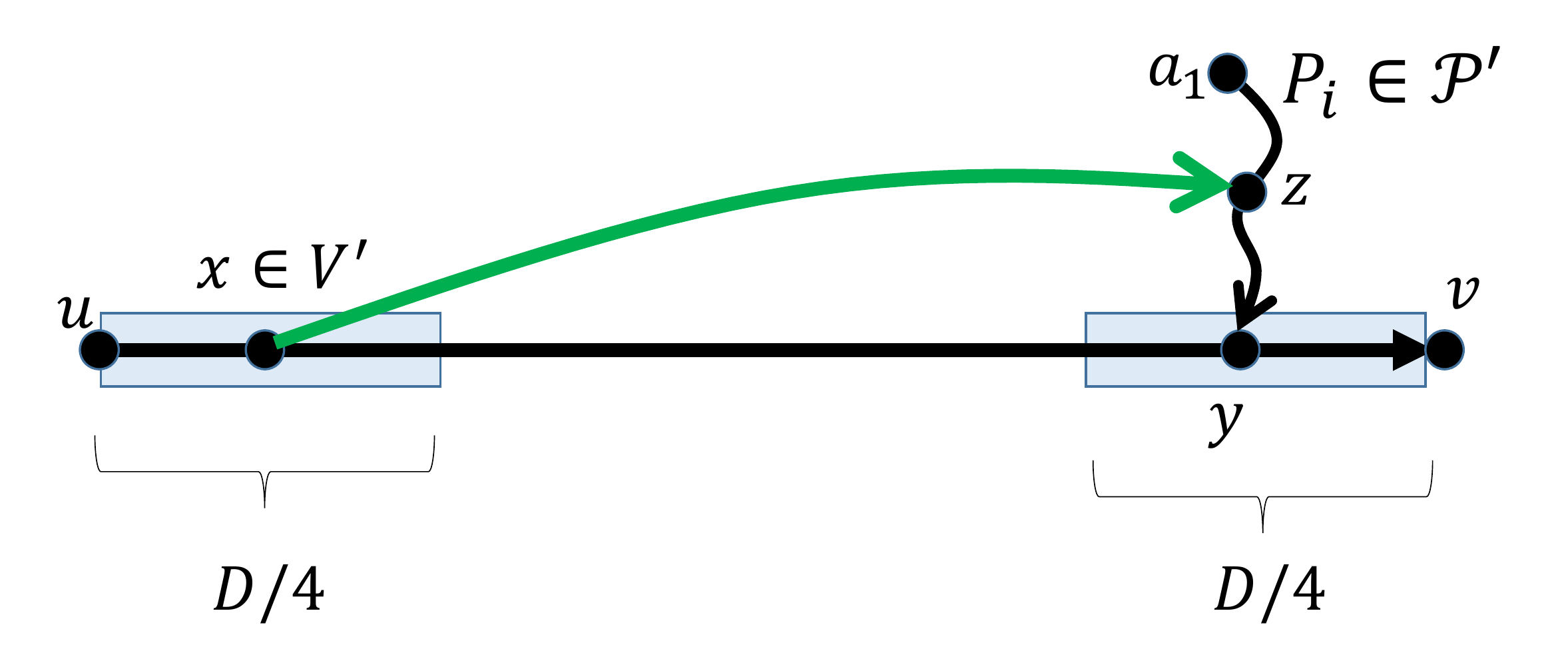}
\caption{\sf An illustration of the diameter bound obtained by Alg. $\ShortcutSmallD$.\label{fig:small-d} 
}
\end{center}
\end{figure}
\subsection{Reducing to Diameter $D=\Omega(n^{1/3})$.}
In this section we explore the complementary regime of $D=\Omega(n^{1/3})$, for which we provide a shortcut set of \emph{sublinear} cardinality. Prior to our work, sublinear shortcuts were obtained only for diameter bound of $\omega(n^{1/2})$. 
\BTHM\label{thm:second_case_mainthm1}
For every $n$-vertex digraph $G(V,E)$ and a parameter $D\geq n^{1/3}$, one can compute in time $\widetilde{O}(n^\omega)$, an edge set $H \subseteq G^{\star}$ of cardinality $|H|=\widetilde{O}((n/D)^{3/2})$ such that w.h.p. for every $(u,v) \in G^{\star}$ it holds that $\dist_{H \cup G}(u,v) \leq D$. 
\ETHM
%
\begin{mdframed}[hidealllines=false,backgroundcolor=gray!00]
\center \textbf{Diameter Shortcutting Algorithm $\ShortcutLargeD$:}
\begin{flushleft}
\textbf{Input:} An $n$-vertex graph $G$ and a diameter bound $D \geq n^{1/3}$. \\
\textbf{Output:} A shortcut set $H \subseteq G^{\star}$ satisfying that $\max_{(u,v)\in G^{\star}}\dist_{G \cup H}(u,v)\leq D$.
\end{flushleft}

\vspace{-2pt}
\begin{enumerate}
\item Let $V' =V[p]$ for $p=\min\{(\sqrt{n}\log n)/D^{3/2}, 1\}$ and set $n'=|V'|$.
\item Let $G'=(V',E')$ where $E'=\{(u,v) ~\mid~ u,v \in V' \mbox{~and~} \dist_{G}(u,v)\leq D^{3/2}/\sqrt{n}\}$.
\item Output $H = \ShortcutSmallD(G',D'=(n')^{1/3} / \log n)$.
\end{enumerate}
\end{mdframed}
%
%
\noindent \textbf{Running Time and Edge Bound.} 
The computation of the graph $G'$ can be done in $\widetilde{O}(n^\omega)$ time by computing a partial transitive order of $G$. In particular, using a logarithmic number of iterations
we compute $(A+I)^k$ where $k=D^{3/2}/\sqrt{n}$, $A$ is the adjacency matrix of $G$, and $I$ is the $n\times n$ identity matrix.  By the Chernoff bound, w.h.p., $n' = O((n/D)^{3/2} \log n)$. 
Applying algorithm $\ShortcutLargeD$ takes $\widetilde{O}((n')^\omega)$ time, w.h.p., by Theorem \ref{thm:first_case_mainthm1}. The number of edges added to $H$ in Step (3) is $\widetilde{O}(n')= \widetilde{O}((n/D)^{3/2})$, w.h.p. 
\\ \\
\noindent \textbf{Diameter bound.} 
For a fixed pair $u,v$, let $P_{u,v}$ be a shortest $u$-$v$ path in $G \cup H$, and assume that $|P_{u,v}| \geq D$. Let $r=D^{3/2}/\sqrt{n}$. 
Partition $P_{u,v}$ into consecutive subpaths of length $r/2$ (except the last subpath, which might be smaller than $r/2$).
Let these subpaths be $P_1,P_2,\ldots,P_q$ where $P_{u,v}=P_1 \circ P_2 \circ \ldots \circ P_q$ and
$|P_1|=|P_2|=\ldots=|P_{q-1}|=r/2$, and furthermore $|P_q|\leq r/2$.
Since every vertex in $V'$ is chosen with probability $p=\min\{(\sqrt{n}\log n)/D^{3/2}, 1\}$, by the Chernoff bound, $P_i \cap V'\neq \emptyset$ for every $1 \leq i \leq q-1$, w.h.p. Let $v_i$ be some some arbitrary sampled vertex, say $v_i$, in $V' \cap V(P_i)$.

The distance between vertex $v_i$ and $v_{i+1}$ in $G$ is at most $r$ for all $1 \leq i \leq q-1$, and hence by Step (2) of the algorithm graph $G'$ contains a path between $v_1$ and $v_{q-1}$. Thus after Step (3) of the algorithm, $G' \cup H$ contains a directed $v_1$-$v_{q-1}$ path $P'$ of length $(n')^{1/3} / \log n$. Since each $G'$-edge corresponds to a path of length at most $r$ in $G$, we get that $P'$ corresponds to a path of length $r \cdot (n')^{1/3} / \log n = O(D)$ in $G \cup H$.  Now the path from $v_1$ to $v_{q-1}$ of length $O(D)$ in $G \cup H$ can be extended to a path of length $O(D)$ between $u$ and $v$ in $G \cup H$ (as there is a path of length at most $r/2$ between $u$ and $v_1$ and a path of length at most $r$ between $u_{q-1}$ and $v$).
This concludes the proof of Theorem \ref{thm:second_case_mainthm1}. See Fig. \ref{fig:large-d} for an illustration.

\begin{figure}[h!]
\begin{center}
\includegraphics[scale=0.35]{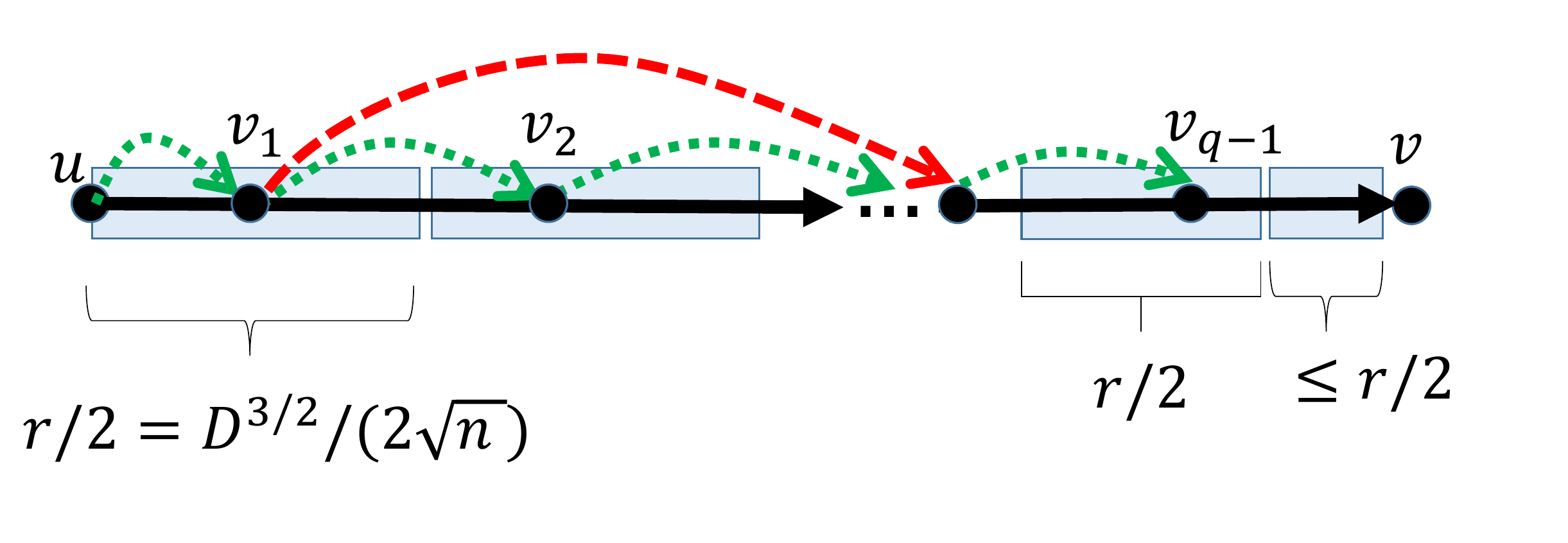}
\caption{\sf An illustration of the diameter bound obtained by Alg. $\ShortcutLargeD$. The green dashed edges correspond to edges of $G'$. The red edge is a shortcut edge in $H$. Since an edge in $G'$ translates into a path of length at most $r$ in $G$, the diameter bound in $G \cup H$ is bounded by $O((n')^{1/3} / \log n \cdot r)$.  \label{fig:large-d} 
}
\end{center}
\end{figure}
We are now ready to complete the proof of Cor. \ref{cor:TC} and provide an improved approximation algorithm for TC-spanners. 
\begin{proof}[Proof of Cor. \ref{cor:TC}]
A transitive reduction of $G$ is a directed graph $G'$ such that $G$ and $G'$ have the same transitive closure, and in addition $G'$ is the smallest graph (by number of edges) that satisfies this condition. Clearly, a $k$-TC spanner must contain at least $O(|E(G')|)$ edges. Our output graph $\widetilde{H}$ is obtained by applying the shortcutting construction algorithm of Theorem \ref{main_theorem_math_ver1} on the graph $G'$ and diameter bound $D=k$. 
This results in a shortcut set $H \subseteq G^{\star}$ satisfying that $G' \cup H$ has diameter at most $k$, thus $H \subseteq G^{\star}$ is a $k$-TC spanner for $G$. We have that $|H|=\widetilde{O}(\max\{n^2 / k^3, n\})$ which sets the desired approximation of $O(n / k^3)$ since we can assume, w.l.o.g. that $G'$ is weakly connected, and thus has at least $n-1$ edges. A similar usage of the transitive reduction has been also employed in \cite{BhattacharyyaGJRW12}. The key contribution is in the application of Theorem \ref{main_theorem_math_ver1} on that graph. 
\end{proof}

\section{$(1+\epsilon)$-Approximate Directed Hopsets}\label{sec:hopsets}
Let $G^{\star}_{\omega}$ be the weighted transitive closure of a weighted digraph $G=(V,E,\omega)$, where each edge $e=(u,v) \in G^{\star}_{\omega}$ is weighted $\dist_G(u,v)$. The weighted transitive closure graph can be computed by applying an APSP algorithm in $G$. 
We first provide a hopset construction for hopbound at most $\beta \leq n^{1/4}$. The cardinality of the shortcut edges obtained for this regime should be compared with Theorem \ref{thm:first_case_mainthm1}. The main difference is that in the weighted hopset construction the diameter threshold is $n^{1/4}$, and in our directed shortcut construction the threshold is improved to $n^{1/3}$. 


\BTHM\label{thm:main_hopset_bound_small_D}
For every $n$-vertex digraph $G=(V,E,\omega)$ with integer weights in $[1,W]$, hopbound $\beta=O(n^{1/4})$ and $\epsilon \in (0,1)$, one can compute an $(\epsilon,\beta)$-hopset $H$ of cardinality $|H|=\widetilde{O}(n^2 \cdot \log (n W)/(\epsilon^2 \beta^3))$, where $W$ is the ratio between the maximum edge weight and the minimum edge weight in $G$. The computation time is $O(n^3 \cdot \poly\log(n W))$. 
\ETHM
Throughout, all edges $(u,v)$ added to the hopset are weighted by $\dist_G(u,v)$, namely, the edge weight in $G^{\star}_{\omega}$.  In contrast to the construction of the (unweighted) diameter shortcut sets, here one \emph{cannot} assume that $G$ is a directed acyclic graph. 
%
%

\paragraph{Nice Path Collection.}
Our algorithm is based on a computation of an \emph{ordered} path collection $\mathcal{Q}=\{P'_1,\ldots, P'_k\}$ with several useful properties\footnote{Some of these properties are crucial for the efficiency of computation, and some are important for the stretch argument.}. A collection of paths $\mathcal{Q}$ is \emph{nice} if it satisfies the following properties: 
\begin{itemize}
\item{(N1)} $P'_i \subseteq G^{\star}_{\omega}$, and all paths in $\mathcal{Q}$ are vertex-disjoint.

\item{(N2)} Every $P'_i \in \mathcal{Q}$ has $\beta/12$ hops.

\item{(N3)} $\len(P'_i)=\dist_G(u_i,v_i)$ where $u_i, v_i$ are the endpoints of $P'_i$, for every $P'_i \in \mathcal{Q}$.

\item{(N4)} $\len(P'_1) \leq \len(P'_2) \leq \ldots \leq \len(P'_k)$.

\item{(N5)} For every $i \in \{1,\ldots, k\}$, it holds that $\len(P'_i)$ is the smallest among all other $(\beta/12)$-hop shortest paths in $G^{\star}_{\omega} \setminus \bigcup_{j=1}^{i-1} V(P'_j)$.

\item{(N6)} For every $u$-$v$ path $P \subseteq G^{\star}_{\omega} \setminus \bigcup_{P'_i \in \mathcal{Q}}V(P'_i)$ with at least $\beta/12$ hops, it holds that $\len(P)> \dist_G(u,v)$. In other words, the set $\mathcal{Q}$ is maximal w.r.t properties (N1-N5).
%
\end{itemize}

\begin{obs}
The set $\mathcal{Q}$ is non-empty iff the graph $G$ contains at least one $u$-$v$ shortest path with at least $\beta$ hops.
\end{obs}
Note that a-priori, it might not be possible to efficiently find a path with at least $\beta'$ number of hops. The nice set $\mathcal{Q}$ can be computed efficiently, however, due to our extra requirement for these paths to be also shortest path in the graph $G^{\star}_{\omega}$. The following theorem is proven in the full version.
\begin{thm}\label{lem:setQ}
Given an $n$-vertex weighted graph $G=(V,E, \omega)$ with integer weights $[1,W]$, one can compute a nice path collection $\mathcal{Q}$ that satisfies properties (N1-N6) in time $O(n^3 \cdot \poly\log(n W))$. 
\end{thm}

\paragraph{The Hopset Algorithm.} Equipped with the computation of the nice path collection $\mathcal{Q}$, we next turn to provide the $(\epsilon,\beta)$ hopset construction for $\beta=O(n^{1/4})$. Upon computing $\mathcal{Q}$, the algorithm first adds to the output set $H$, the hopset edges $H_i=H(P'_i)$ for each path $P'_i \in \mathcal{Q}$. In contrast to the diameter shortcutting setting, here $H(P'_i)=G^{\star}_{\omega}[V(P'_i)]$ will have a quadratic size of $O(|V(P'_i)|^2)$ edges, rather than linear size in $V(P'_i)$, as obtained by Lemma \ref{line_shortcut_lem1}. The reason for this has to do with the fact that $G$ is no longer assumed to be a DAG. 

Next, the algorithm computes a more refined path collection $\mathcal{P}$ by partitioning the \emph{vertices} of each path $P'_i \in \mathcal{Q}$ into $O(1/\epsilon)$ subpaths $P'_{i,1}, \ldots, P'_{i,k_i}$, each of length at most $\epsilon \cdot \len(P'_i)$. Importantly, these subpaths do not cover all edges in $P'_i$. In addition, some of the $P'_{i,j}$ paths might consist of a singleton vertex. The important properties of these subpaths are that (i) they cover all vertices in $P'_i$, i.e., $\bigcup_{j} V(P'_{i,j})=V(P'_i)$ and that (ii) $\len(P'_{i,j})\leq \epsilon \cdot \len(P'_i)$. Altogether, $|\mathcal{P}|=O(1/\epsilon)\cdot |\mathcal{Q}|$.

The computation of the $P'_{i,j}$ subpaths (to be added to $\mathcal{P}$) is done by applying a simple iterative procedure on $P'_i=[a_1,\ldots,a_\ell] \in \mathcal{Q}$. Let $a_{j_1}$ be the largest index on $P'_i$ satisfying that $\len(P'_i[a_1,a_{j_1}])\leq \epsilon \cdot \len(P'_i)$. Set $P'_{i,1}=P'_i[a_1,a_{j_1}]$.  Note that it might be the case that $a_{j_1}=a_1$, and in such a case, $P'_{i,1}=\{a_1\}$. The algorithm then proceeds on the remaining path $P'_i[a_{j_1+1}, a_k]$ to define $P'_{i,2}$. Note that the edge $(a_{j_1}, a_{j_1+1})$ will not appear in any of the output subpaths. 
This continues for $O(1/\epsilon)$ steps. Overall, we get a partitioning of the $P'_i$ vertices into $O(1/\epsilon)$ subpaths and singleton vertices (we sightly override notation and refer to these singleton vertices as paths as well).  

The algorithm then proceeds with the refined path collection $\mathcal{P}$.
In a similar manner to the construction of the shortcutting set (of Alg. $\ShortcutSmallD$), it samples a subset of paths $\mathcal{P}' \subseteq \mathcal{P}$, and a subset of vertices $V' \subseteq V$. It then adds a collection of $ O(\log (n W) / \epsilon)$ weighted edges between every sampled vertex $v \in V'$ and every sampled path $P=[u_1,\ldots, u_k] \in \mathcal{P}'$. Whereas in the diameter shortcutting algorithm, it was sufficient to add a single edge $e(v,P)$, here the algorithm adds a \emph{subset} $E(v,P)$ of $ O(\log (n W) / \epsilon)$ weighted edges defined as follows. Let $u_i$ be the first vertex on $P$ (of minimal index $i$) satisfying that $v \leadsto_G u_i$. Then include the weighted edge $(v,u_i)$ in $E(v,P)$. The additional edges from $v$ to $V(P)$ are added to $E(v,P)$ in a sequential manner. For $j\geq i$, let $u_j$ be the last vertex on $P$ such that $(v,u_j)$ was added to $E(v,P)$. Then, the set 
$E(v,P)$ contains $(v,u_{j'})$ for the first vertex in $P[u_{j}, u_{k}]$ satisfying that
$(1+\epsilon) \dist_{G}(v, u_{j'})< \dist_{G}(v, u_j)~.$
Note that 
\begin{equation}\label{eq:epsilon_reduction1}
|E(v,P)|=O(\log_{1+\epsilon} (n W)) = O(\log (n W) / \epsilon).
\end{equation}
The construction algorithm of the $(\epsilon,\beta)$ hopset is formalized next. 

\begin{mdframed}[hidealllines=false,backgroundcolor=gray!00]
\center \textbf{Algorithm $\HopsetSmallHop$:}
\begin{flushleft}
\textbf{Input:} An $n$-vertex weighted $G=(V,E,\omega)$ with weights in $[1,W]$, hopbound $\beta\leq n^{1/4}$ and an $0<\epsilon$. \\
\textbf{Output:} An $(\epsilon,\beta)$ hopset $H \subseteq G^{\star}_{\omega}$.
\end{flushleft}

\vspace{-10pt}
\begin{enumerate}
\item Compute the weighted transitive closure $G^{\star}_\omega$, and a nice path collection $\mathcal{Q}$ by applying Theorem \ref{lem:setQ}. 

%
\item For every $P'_i \in \mathcal{Q}$, add the edges $H(P'_i)=G^{\star}_\omega[V(P'_i)]$ to $H$.
   
\item For every $P'_i \in \mathcal{Q}$, the vertices of $P'_i$ are partitioned into $k_i = O(1/\epsilon)$ subpaths $P'_{i,1},\ldots, P'_{i,k_i}$ each of length at most $\epsilon \cdot  \ell_i$ where $\ell_i=\len(P'_i)$. 
Let $\mathcal{P}=\{P'_{i,1},\ldots, P'_{i,k_i} ~\mid~ P'_i \in \mathcal{Q}\}$ be the collection of $O(k/\epsilon)$ paths. 
%
\item Let $V' \subset V$ be a random sample of $O(n\log n/\beta)$ vertices, and let $\mathcal{P}' \subset \mathcal{P}$ be a random sample of $O(|\mathcal{P}|\log n/\beta)$ paths\footnote{Obtained by sampling each $P \in \mathcal{P}$ into $\mathcal{P}'$ with probability $\Theta(\log n/\beta)$.}.
\item For every $v \in V'$ and $P_{j} \in \mathcal{P'}$, add the weighted edges $E(v,P_j)$ to the hopset $H$ .
\end{enumerate}
\end{mdframed}

\paragraph{Running time.} 
By Theorem \ref{lem:setQ}, the computation of the nice path collection $\mathcal{Q}$ runs in $O(n^3 \poly\log(n W))$ time.  We next analyze the running time of the remaining steps of Algorithm $\HopsetSmallHop$. Step (1) runs in time $O(n^3 \poly\log(n W))$. Step (2) runs in time $O(n^2)$, Steps (3,4) run in $O(n)$ time. Finally,  Step (5) runs in time $O(n^2)$ (as we can use the fact that  $G^{\star}_{\omega}$ is already computed during Step 1).
We conclude that the time complexity of Algorithm  $\HopsetSmallHop$ is $O(n^3 \cdot \poly\log(n W))$.

%

\paragraph{Size Bound.} 
We start by bounding the number of edges in $H$. By properties (N1) and (N2), $|\mathcal{Q}|=O(n/\beta)$. 
Step (2), adds $O(\beta^2)$ edges to $H$ for each path $P'_i \in \mathcal{Q}$. Therefore, it adds a total of 
$O(n \beta) = O(n^2/\beta^3)$ edges to $H$, where the last equality follows from the fact that $\beta = O(n^{1/4})$. To bound the number of edges in Step (3), observe that $|\mathcal{P}|= O(n/(\epsilon \cdot \beta))$ and therefore $|\mathcal{P}'|=O(n\log n/ (\epsilon \cdot\beta^{2}))$ w.h.p. Since $|V'|=O(n \log n /\beta)$, the total number of edges added in Step (3) to $H$ is bounded by
$$|V'|\cdot |\mathcal{P}'|\cdot O(\log( n W)/\epsilon)=\widetilde{O}(n^2 \cdot \log(n W)/(\epsilon^2 \beta^3))~.$$

\paragraph{Stretch and Hopbound.} Let $M = \bigcup_{P'_i \in \mathcal{Q}}V(P'_i)$ be the vertices lying on the paths of $\mathcal{Q}$ (and thus also on the paths of $\mathcal{P}$), and let $R=V(G)\setminus M$ be the remaining vertices. Let $H'=\bigcup_{P'_i \in \mathcal{Q}} H(P'_i)$ be the set of edges in $H$ added in
Step (2) of the algorithm, where $H(P'_i)=G^{\star}_{\omega}[V(P'_i)]$. For a fixed pair $u,v  \in V$, let $P_{u,v}$ be a shortest $u$-$v$ path in $G \cup H'$ with the \emph{minimal} number of hops among all other $u$-$v$ shortest paths. Assume that $P_{u,v}$ has at least $\beta$ edges (hops), as otherwise the claim follows immediately. 
%
We need the following simple observation
\BL\label{lem:reminder_of_the_graph_bound1}
$P_{u,v}$ contains at most $\beta/12$ vertices from $R$.
\EL
\BPF
If $P_{u,v}$ contains a vertex set $S \subseteq R$ such that $|S| = \beta/12 + 1$ then this set induces a $(\beta/12)$-hop shortest path in $G^{\star}_\omega$, leading to a contradiction with property (N6). That is, by the maximality of the set $\mathcal{Q}$, any remaining path of $\beta/12$ hops is not a shortest path in $G^{\star}_{\omega}$. 
%
%
\EPF
Next, let $S$ and $S'$ be the $(\beta/3)$-hop prefix and suffix of $P_{u,v}$, respectively.
\BL\label{lem:unique_elements_bound1}
$S'$ contains at most $2$ representative vertices from each path $P'_i \in \mathcal{P}'$. I.e., $|V(S) \cap V(P'_i)|\leq 2$ for every $P'_i \in \mathcal{P}'$.
\EL
\BPF
As $H'=\bigcup_{P'_i \in \mathcal{P}'} H(P'_i)$ and for all $i$ we have that $ H(P'_i)$ contains all the edges of the subgraph of $G^{\star}_{\omega}[V(P'_i)]$. In particular if $P'_i = [v_1,v_2,\ldots,v_{\beta/12}]$, then for any $1 \leq j < k \leq \beta/12$, the set $ H(P'_i)$ contains an edge $(v_j,v_k)$ with weight $\dist_G(v_j,v_k)$. 
The claim follows now from the fact that $P_{u,v}$ was chosen to be a shortest $u$-$v$ path in $G \cup H'$ with the minimal number of hops. Suppose that we had $3$ representatives $a,b,c \in P'_i$ in $S'$ which appears in $P_{u,v}$ in this order, then we can remove vertex $b$ and get a path from $u$ to $v$ of the same length but with less hops, thus getting a contradiction.
\EPF

\BL\label{lem:weighted_sets_bound1}
There is a vertex subset $S'' \subseteq V(S') \bigcap M$ such that:
\begin{itemize}
\item $|S''|=\beta/24$,
\item each $P'_i \in \mathcal{Q}$ contains at most one vertex from $S''$, i.e., $|V(P'_i) \cap V(S'')|\leq 1$,
\item letting $P_s$ be the unique path in $\mathcal{P}$ containing $s \in S''$, it holds that $\len(P_s)\leq \epsilon \cdot \len(P_{u,v})$.
\end{itemize}
\EL
\BPF
By Lemma \ref{lem:reminder_of_the_graph_bound1}, there is a set $S_1 \subseteq V(S') \cap M $ such that $|S_1| \geq \beta/3 - \beta/12 = \beta/4$.
By Lemma \ref{lem:unique_elements_bound1}, there is a set $S_2 \subseteq S_1$ such that $|S_2| \geq |S_1|/2 = \beta/8$ and in addition, each vertex in $S_2$ is contained in some $P'_i$ such that $|V(P'_i)\cap S_2|=1$.
Let $\mathcal{P}'' \subseteq \mathcal{P}'$ be a subset of $\beta/8$ paths that intersect $S_2$. 
That is, for every path $P \in \mathcal{P}''$, it holds that $|V(P)\cap S_2|=1$.

Assume without loss of generality that $\mathcal{P}'' = \{Q_{i,1},Q_{i,2},\ldots,Q_{i,\beta/8}\}$ is ordered by the ordering given in the ordered path collection $\mathcal{Q}$. That is, $Q_{i,j}$ appears in $\mathcal{Q}$ before $Q_{i,j+1}$ for every $j \in \{1,\ldots, \beta/8-1\}$. 
Let $\ell$ be the index of the path $Q_{i,\beta/24}$ in $\mathcal{Q}$. That is, the $\ell$'th path in $\mathcal{Q}$, namely, $P'_{\ell}$ is the path $Q_{i,\beta/24}$. The path $P_{u,v}$ then contains a vertex subset $S_3 \subseteq S_2$ of cardinality $\beta/12$. Specifically, each vertex in $S_3$ appears on a unique path in $\{Q_{i,\beta/24+1}, \ldots, Q_{i,\beta/8}\}$ (and each path in this set intersects a unique vertex in $S_3$). Note that $S_3$ induces a shortest path $P_3$ of exactly $\beta/12$ hops in $G^{\star}_\omega$. By property (N5), it holds that $\len(P'_{\ell+1})\leq \len(P_3)$. See Fig. \ref{fig:hopset-fig} for an illustration.
%
%
%
%
%
%
%
Therefore, 
we can conclude that $\len(P_{u,v})\geq \len(P_3)\geq \len(Q_{i,\beta/24 + 1})$. By combining with the ordering of  
$\mathcal{P}''$ and property (N4) of the nice set $\mathcal{Q}$, we have that 
$$\len(Q_{i,1})\leq \len(Q_{i,2})\leq \ldots \leq \len(Q_{i,\beta/24 + 1})\leq \len(P_{u,v})~.$$
The desired set $S''$ is given by $S''=S_2 \cap \bigcup_{j=1}^{\beta/24} V(Q_{i,j})$. 
%
%
 
Finally, by Step (3) of the algorithm, we have every vertex in $S''$ lies on some path $P_{i',j'} \in \mathcal{P}$ where $P_{i',j'} \subseteq P'_{i'}$ for $P'_{i'} \in \mathcal{Q}$. In addition, $\len(P_{i',j'})\leq \epsilon \cdot \len(P'_{i'})$. Since $\len(P'_{i'})\leq \len(P_{u,v})$, we have that $\len(P_{i',j'})\leq \epsilon \cdot \len(P_{u,v})$ as desired. 
\EPF
We turn to complete the stretch argument on the $\beta$-hop distances in $G \cup H$. The next claim provides a multiplicative stretch of $(1+2\epsilon)$. Providing a stretch of $(1+\epsilon)$ can be done by applying the algorithm with $\epsilon'=\epsilon/2$
\begin{clm}\label{cl:stretch-hopset}
$\dist_{G \cup H}^{\beta}(u,v) \leq (1+2\epsilon) \len(P_{u,v})$.
\end{clm}
\begin{proof}
By Lemma \ref{lem:weighted_sets_bound1} we have that $S''$ contains representatives from $\Omega(\beta)$ paths in $\mathcal{P}$, and thus  w.h.p. $S''$ contains at least one representative vertex $x$ on some $P_i=[a_1,\ldots, a_k]$ in $\mathcal{P}'$ (as each path in $\mathcal{P}$ is chosen with probability $O(\log n/\beta$) in Step 4 of the algorithm). In addition, $S$ contains at least one vertex $y \in V'$ w.h.p. (as each vertex in $S$ is chosen with probability $O(\log n/\beta)$ in Step 4 of the Algorithm, and $|S| = \Omega(\beta)$). Therefore the set $E(y,P_i)$ is in $H$.
By definition, $E(x,P_i)$ contains some edge $(x,z)$ for $z \in P_i[a_1,y]$ and 
\begin{equation}\label{binary_bound1}
(1+\epsilon) d_G(x,y)\geq \dist_{G}(x,z)
\end{equation}
Therefore $G \cup H$ contains a $u$-$v$ path $P'_{u,v}=P_{u,v}[u, x] \circ (x,z) \circ P_i[z,y] \circ P_{u,v}[y,v]$ with at most $\beta$ hops. Thus, the $\beta$-hop $u$-$v$ distance in $G \cup H$ can be bounded by:
\begin{eqnarray}
\dist_{G \cup H}^{\beta}(u,v) &\leq& \len(P'_{u,v}) \leq \len(P_{u,v}[u, x])+\dist_G(x,z)+\len(P_i[z,y])+\len(P_{u,v}[y,v]) \notag
\\& \leq &
\len(P_{u,v}[u, x])+(1+\epsilon) \dist_G(x,y) + \len(P_i[z,y]) +\len(P_{u,v}[y,v])  \label{using_the_binary_bound}
\\& \leq & 
\len(P_{u,v}[u, x])+(1+\epsilon) \dist_G(x,y) + \epsilon \cdot  \len(P_{u,v}) +\len(P_{u,v}[y,v]) \label{tinylengthbound1}
\\&\leq&
(1+2\epsilon) \len(P_{u,v}) \notag,
\end{eqnarray}
where Inequality (\ref{using_the_binary_bound}) follows from Inequality  (\ref{binary_bound1}), and  Inequality (\ref{tinylengthbound1}) follows from Lemma \ref{lem:weighted_sets_bound1}. See Fig. \ref{fig:hopset-fig} for an illustration.
\end{proof}
This completes the proof of Theorem \ref{thm:main_hopset_bound_small_D}.
\begin{figure}[h!]
\begin{center}
\includegraphics[scale=0.35]{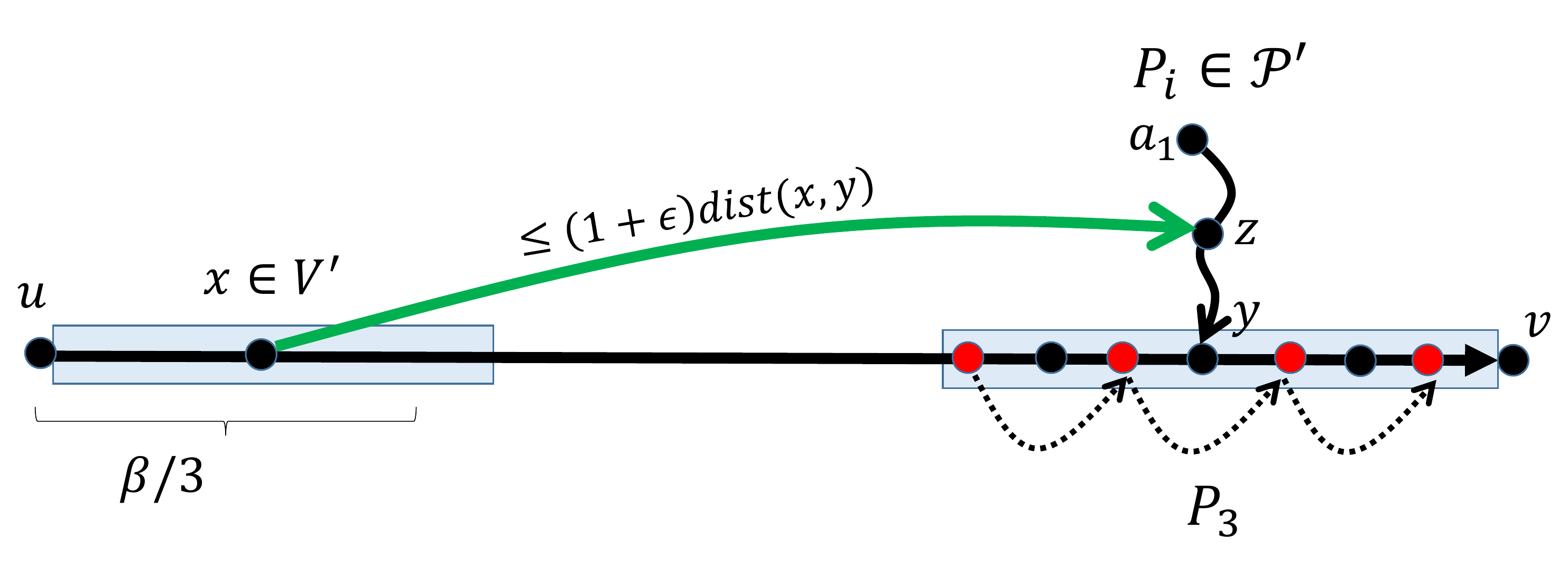}
\caption{\sf An illustration to the approximate $\beta$-hop path obtained by Alg. $\HopsetSmallHop$ (useful for the proofs of Lemma \ref{lem:weighted_sets_bound1} and Claim \ref{cl:stretch-hopset}).
Shown is a $u$-$v$ shortest path in $G \cup H'$. The $\beta/3$-length suffix $S$ contains w.h.p. a sampled vertex $x$ from $V'$. The $\beta/3$-length suffix $S'$ contains unique representatives from $\Omega(\beta)$ paths in $\mathcal{Q}$. The red vertices corresponds to the set $S_3$ (in the proof of Lemma \ref{lem:weighted_sets_bound1}) which induces a $\beta/12$-hop path $P_3$. \label{fig:hopset-fig} 
}
\end{center}
\end{figure}
\noindent In a similar manner to the proof of Theorem \ref{thm:second_case_mainthm1}, we can use again a sampling procedure to obtain approximate hopsets with hopbound $\beta=\Omega(n^{1/4})$. 
\BTHM\label{thm:main_hopset_bound_large_D}
For every $n$-vertex positively weighted directed graph $G$, $\beta=\Omega(n^{1/4})$ and $\epsilon \in (0,1)$, one can compute a weighted edge set $H \subseteq G^{\star}_{\omega}$ of cardinality $|H|=\widetilde{O}\left( \left(\frac{n}{\beta}\right)^{5/3} \epsilon^{-2} \log (n W) \right)$ such that for every $u,v \in V(G)$ it holds that:
\begin{equation} \label{eq:shortcut-eq}
\dist_{G}(u,v)\leq \dist_{H \cup G}^{(\beta)}(u,v)\leq (1+\epsilon)\dist_{G}(u,v)~.
\end{equation}
\ETHM

\def\APPENDLARGEHOP{
We present Algorithm $\HopsetLargeHop$ for computing an $(\epsilon,\beta)$ hopset given hopbound $\beta\geq n^{1/4}$ and $\epsilon \in (0,1)$. 
\begin{mdframed}[hidealllines=false,backgroundcolor=gray!00]
\center \textbf{Algorithm $\HopsetLargeHop$:}
\begin{flushleft}
\textbf{Input:} An $n$-vertex weighted $G=(V,E,\omega)$ with weights in $[1,W]$, hopbound $\beta\geq n^{1/4}$ and $\epsilon>0$. \\
\textbf{Output:} An $(\epsilon,\beta)$ hopset $H \subseteq G^{\star}$.
\end{flushleft}

\vspace{-2pt}
\begin{enumerate}
\item Let $V' \subset V$ be a random sample of $n' = O((n/\beta)^{4/3} \log n)$ vertices of $G$.

\item Let $h(u,v)$ be the minimal number of hops\footnote{That is, for $h=h(u,v)$ it holds that $\dist^{(h)}_G(u,v)=\dist_G(u,v)$ and $\dist^{(h-1)}_G(u,v)>\dist_G(u,v)$.} over all $u$-$v$ shortest path in $G$.
Define $G'=(V',E')$ where 
$$E'=\{(u,v) ~\mid~ u,v \in V' \mbox{~and~} h(u,v)\leq  \beta^{4/3}/n^{1/3} \}$$ 
and the weight of edge $(u,v)$ in $E'$ is given by $\dist^{(h)}_G(u,v)$ for $h=\beta^{4/3}/n^{1/3}$.
\item Output $H = \HopsetSmallHop(G',\beta'=(n')^{1/4} / \log n,\epsilon)$.
\end{enumerate}
\end{mdframed}
Now we describe the time complexity of each step in the algorithm above.

\paragraph{Running Time and HopBound.} 
We assume that the weight are integers in the range $[1,W]$ for the running time analysis.
The computation of the graph $G'$ can be done in $\widetilde{O}(n^3)$ time by computing the APSP of $G$ (technically the Min-Sum Matrix Product for a logarithmic number of times).
Applying Algorithm $\HopsetSmallHop$ takes $O((n')^3 \cdot \poly\log(n W))$ time by Theorem \ref{thm:main_hopset_bound_small_D}. 
The number of edges added to $H$ in step $3$ is $\widetilde{O}((n')^{5/4} \epsilon^{-2} \log(n' \cdot W) )= \widetilde{O}((n/\beta)^{5/3} \epsilon^{-2} \log(nW))$. 

\paragraph{Diameter bound.} 
For a fixed pair $u,v$, let $P_{u,v}$ be a shortest $u$-$v$ path in $G \cup H$ with the minimal number of hops $h(u,v)$. We assume that $P_{u,v}$ has at least $\beta$ edges (hops), for otherwise we are done. 
Let $r=\beta^{4/3}/n^{1/3}$. 
Partition $P_{u,v}$ into consecutive subpaths of $r/2$ hops (except the last subpath which can be smaller than $r/2$).
Let these subpaths be $P_1,P_2,\ldots,P_q$ where $P_{u,v}=P_1 \circ P_2 \circ \ldots \circ P_q$ and
$|P_1|=|P_2|=\ldots=|P_{q-1}|=r/2$ and furthermore $|P_q|\leq r/2$.
Since every vertex in $V'$ is chosen with probability $n' /n$, by the Chernoff bound it holds that, w.h.p., $V(P_i) \cap V'\neq \emptyset$ for every $1 \leq i \leq q-1$. Let $v_i$ be some some representative vertex in $V' \cap V(P_i)$.  We then have that $\dist_G(v_i, v_{i+1})=\dist^{(r)}_G(v_i, v_{i+1})$ for every $1 \leq i \leq q-1$. 
Therefore, by Step (2) of Algorithm $\HopsetLargeHop$, the graph $G'$ contains a directed path from $v_1$ and $v_{q-1}$ of length $\dist_G(v_1,v_{q-1})$.

Therefore, after Step (3) of the algorithm, $G' \cup H$ contains a directed $v_1$-$v_{q-1}$ path $P'$ of at most $(n')^{1/4} / \log n$ hops and $\len(P')\leq (1+\epsilon)\dist_G(v_1,v_{q-1})$. Since each $G'$-edge corresponds to a path of at most $r$ hops in $G$, we get that $P'$ corresponds to a path of at most $r \cdot (n')^{1/4} / \log n  = O(\beta)$ hops in $G \cup H$.  

Next, observe that the $O(\beta)$-hop path $P' \subseteq G \cup H$ can be extended into an $O(\beta)$-hop $u$-$v$ path $P''$ in $G \cup H$. This path is given by
$$P''=P_{u,v}[u,v_1] \circ P' \circ P_{u,v}[v_{q-1},v]~.$$
Since $\len(P')\leq (1+\epsilon)\dist_G(v_1,v_{q-1})$, it holds also that $\len(P'')\leq (1+\epsilon)\dist_G(u,v)$. 
 This concludes the proof of Theorem \ref{thm:main_hopset_bound_large_D}.

When restricting the number of edges in the hopset to $\widetilde{O}_{\epsilon}(n)$, we get hopbound of $\beta=n^{2/5}$. This should be compared with the previous hopbound of $O(\sqrt{n})$. 
\BCR
For every $n$-vertex positively weighted directed graph $G$, $\beta=n^{2/5}$ and $\epsilon \in (0,1)$, one can compute a weighted edge set $H \subseteq G^{\star}_{\omega}$ of cardinality $|H|=\widetilde{O}\left( n \epsilon^{-2} \log (n W) \right)$ such that for every $u,v \in V(G)$ it holds that:
\begin{equation} \label{eq:shortcut-eq}
\dist_{G}(u,v)\leq \dist_{H \cup G}^{(\beta)}(u,v)\leq (1+\epsilon)\dist_{G}(u,v)~.
\end{equation}
\ECR
}%


\section{Lower Bounds for Directed Shortcuts}\label{pettie_appendix}
In this section, we generalize the following lower bound result of \cite{HuangP18}, and provide the proof for Theorem \ref{pettie_lower_bound}. Our goal is to provide a smooth tradeoff between the number of shortcut edges and the diameter obtained. 
\BTHM\label{pettie_bound1}[\cite{HuangP18}]
There exists an $n$-vertex digraph $G=(V,E)$ such that for any shortcut set $E' \subseteq G^{\star}$ with $|E'|=O(n)$, the graph $G \cup E'$ has diameter $\Omega(n^{1/6})$.
\ETHM
The lower bound graph $G$ of Theorem \ref{pettie_bound1} has been shown in \cite{HuangP18} to satisfy the following properties.
\begin{enumerate}
\item Let $\Delta_{\inn}, \Delta_{\outt}$ be the maximum in-degree (resp., out-degree) in $G$. 
Then $\Delta_{\inn}, \Delta_{\outt}= O(n^{1/6})$. 
\item $G$ contains a collection $\mathcal{P}$ of $\Theta(n)$ paths, each of length $D=\Theta(n^{1/6})$.
\item Each path $P \in \mathcal{P}$ is the only directed path in $G$ between its endpoints.
\item Every two paths $P, P' \in \mathcal{P}$ intersect on at most one vertex.
\item Each vertex of $G$ is contained in at most $\max(\Delta_{\inn},\Delta_{\outt}) = O(n^{1/6})$ paths of $\mathcal{P}$. 
\item The graph $G$ is a DAG (directed acyclic graph).
\end{enumerate}
Notice that property (5) follows from properties (1) and (4), as by property (4) the paths of $\mathcal{P}$ are edge-disjoint. We use the following key property from \cite{HuangP18}:
\begin{lem}[Slight Restatement of Lemma 2.3 from \cite{HuangP18}]\label{lem:LB-aux}
Let $E'$ be a shortcut set for $G=(V,E)$. Any shortcut edge can only be useful for the endpoints of at most one path in $\mathcal{P}$. Consequently, if the diameter of $G'$ is strictly less than $D$, then $|E'|\geq |\mathcal{P}|$. 
\end{lem}
Given an integer parameter $k \geq 1$, let $G_k$ be a graph obtained from graph $G$ by replacing each vertex $v_i$ of $G$ by a directed path $Q_i=[u^i_1,u^i_2,\ldots,u^i_{k+1}]$ of length $k$.
Formally, letting $V(G)=\{v_1,v_2,\ldots,v_n\}$, then for every $v_i \in V(G)$, let $Q_i=[u^i_1,u^i_2,\ldots,u^i_{k+1}]$ be a directed path of length $k$. Thus $V(G_k)=\{ u^i_j ~\mid~ i \in \{1,\ldots, n\}, j \in \{1,\ldots, k\}\}$,
and $|V(G_k)|=n(k+1)$ vertices. In addition, for every $(v_i,v_j) \in E(G)$, the edge set $E(G_k)$ contains the edge $(u^i_{k+1},u^j_1)$. Consequently, $G_k$ contains a collection $\mathcal{P}'$ of directed paths where 
$|\mathcal{P}'|=|\mathcal{P}|=\Theta(n)$. Specifically, for every path $P_j=[v_{j,1}, \ldots, v_{j,D+1}] \in \mathcal{P}$, there is a corresponding path $P'_j=Q_{j,1}\circ \ldots \circ Q_{j,D+1}$ in $\mathcal{P}'$ formed by concatenating the $k$-length paths of the $P_j$'s vertices. The length of each path $P'_j$ is thus $D \cdot k$. See Fig. \ref{fig:LB} for an illustration.
We next show that the diameter lower bound on $G$ can be translated into a diameter lower bound in $G_k$.
Interestingly, a key property of the graph $G$ that we exploit is that $\Delta_{in},\Delta_{out}=O(D)$. As we will see, this plays an important role in our argument. 
\BL\label{inner_outer_lem1}
Let $E'$ be a shortcut set of edges for graph $G_k$. If the diameter of $G_k\cup  E'$ is strictly less than $D \cdot k/2$, then $|E'| = \Omega(|\mathcal{P}'|)$.
\EL
\BPF
We classify the edges in $E'$ into two types. An edge $e=(u,v) \in E'$ is an \emph{outer edge} if it connects
vertices in \emph{distinct} $k$-length paths, i.e., $u\in Q_i,v \in Q_j$ where $G$-vertices $v_i,v_j$ lie on the \emph{same} path in $\mathcal{P}$.  An edge $e=(u,v)$ is \emph{internal} if $u$ and $v$ belong to the same $k$-length path, say $Q_j$ corresponding to vertex $v_j \in G$. Let $E'_{out},E'_{in}$ be the subset of outer (resp., internal) edges in $E'$. See Fig. \ref{fig:LB} for an illustration.
Notice that by properties (3) and (6), if an edge in $E'$ is neither outer nor internal (i.e., an edge whose endpoints are on distinct paths in $\mathcal{P}'$), then it does not shortcut any of the paths in $\mathcal{P}'$. We can therefore assume, w.l.o.g., that $E'$ consists of only outer and internal edges, i.e., $E'=E'_{out} \cup E'_{in}$.

By Lemma \ref{lem:LB-aux}, it holds that each outer edge can be useful to shortcut at most one path in $\mathcal{P}'$. This holds since each outer edge corresponds to an edge in $G^{\star}$ (the transitive closure of $G$). If at least half of the paths in $\mathcal{P}'$ have some outer edge in $E'_{out}$, then clearly, $|E'|\geq |\mathcal{P}'|/2$, and we are done. Assume otherwise, and let $\mathcal{P}''\subseteq \mathcal{P}'$ be the collection of paths that have no outer edge in $E'_{out}$. 
A single internal edge reduces the length of a path in $\mathcal{P}''$ by a most $k$. Therefore, a path in $\mathcal{P}''$ must have at least $D/2$ internal edges in order to reduce its length to below $Dk/2$. Letting $E'_i$ be the set of internal edges in $E'_{in}$ corresponding to a path $P_i \in \mathcal{P}''$, we have that $|E'_{i}|\geq D/2$. 
By property (5), each internal edge (that corresponds to a unique vertex in $G$) can reduce the length of at most $\max\{\Delta_{\inn}, \Delta_{\outt}\}=O(n^{1/6})$ paths in $\mathcal{P}'$. 
At this point, we exploit the fact that $\max\{\Delta_{\inn}, \Delta_{\outt}\}=O(D)$.
Hence the total number of inner edges $|E'_{in}|$ satisfy
$$|E'_{in}| \geq \frac{\sum_{P_i \in \mathcal{P}''} |E'_{i}|}{O(n^{1/6})}\geq \frac{|\mathcal{P}''|\cdot D/2}{O(n^{1/6})}\geq c|\mathcal{P}''|~,$$
for some constant $c$. Now as $|\mathcal{P}''| \geq  |\mathcal{P}'|/2$ we have that $|E'|=|E'_{out}|+|E'_{in}|=\Omega(|\mathcal{P}'|)$, as required. 
The claim follows.
%
%
%
%
\EPF
We are now ready to complete the proof of Theorem \ref{pettie_lower_bound}. 
The graph $G_k$ has $n \cdot (k+1)$ vertices, and by Lemma \ref{inner_outer_lem1} we need to add at least $\Omega(|\mathcal{P}'|) = \Omega(n)$ edges from the transitive closure of $G_k$ in 
order to reduce that diameter of $G_k$ to $D' = D \cdot k/2$.
Hence $k= 2D'/D$ and $|V(G_k)|= n \cdot (k+1) = (2D' + D) \cdot n / D$. Thus
$n/D=|V(G_k)|/(2D'+D)$.  Plugging in $D=\Theta(n^{1/6})$, we have that $\Theta(n^{5/6})=\Theta(|V(G_k)|/D')$. 
We therefore have that
\[\Theta(n) = \Theta\left(\left(\frac{|V(G_k)|}{D'}\right)^{6/5}\right)~.
\]
Theorem \ref{pettie_lower_bound} follows.


\begin{figure}[h!]
\begin{center}
\includegraphics[scale=0.35]{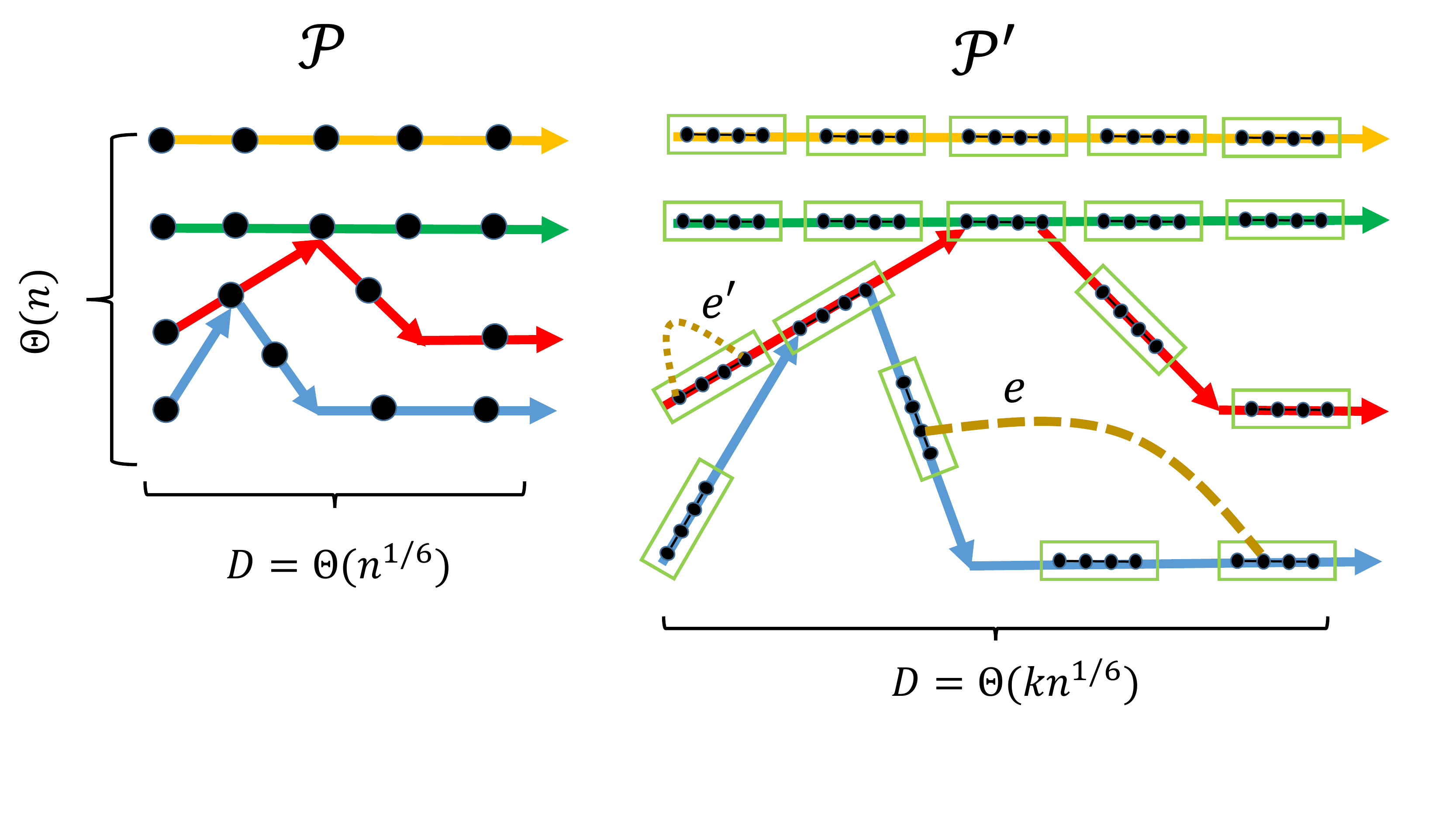}
\caption{\sf An illustration of the lower bound graph construction of Theorem \ref{pettie_lower_bound}.
Right: Illustration for the lower bound graph $G$ of Lemma \ref{pettie_bound1} from \cite{HuangP18}. 
Shown is the collection of the path collection $\mathcal{P}$, each path is of length $D=\Theta(n^{1/6})$. 
Left: The modified graph $G_k$ obtained by replacing each $G$-vertex with an $k$-length directed path. Consequently, $G_k$ consists of $(k+1) n$ vertices, and a collection $\mathcal{P}'$ of $|\mathcal{P}|$ paths, each of length $k D$. 
The dotted edge $e'$ is an internal shortcut edge, and the dashed edge $e$ is an outer shortcut edge. 
\label{fig:LB} 
}
\end{center}
\end{figure}
\paragraph{Acknowledgment.} We are grateful for the very useful comments of SODA 2022 reviewers. The second author is grateful to Seth Pettie that sparked her interest in diameter reducing shortcuts.

\bibliographystyle{alpha}
\bibliography{thesis}

\newpage

\begin{appendix}

\newpage
\section{Missing Proofs for Approximate Hopset Constructions}

\subsection{Efficient Implementation of Alg. $\HopsetSmallHop$ (Proof of Theorem \ref{lem:setQ})}\label{app:hopset-efficient}

In this section, we prove Theorem \ref{lem:setQ} for computing a nice path collection $\mathcal{Q}$. 
Our algorithm is based on the dynamic APSP algorithm of Thorup (\cite{Thorup04}). In particular, our algorithm needs the decremental dynamic APSP whose properties are summarized below. 
\begin{thm}\label{thm:dyn-thorup}[Dynamic APSP Algorithm, \cite{Thorup04}]
Given a digraph $G=(V,E,\omega)$ with possibly negative weights, there is a preprocessing algorithm that in $\widetilde{O}(n^3)$ time computes the APSP matrix $A$ of $G$. In addition, there is an update algorithm that upon any removal of a vertex $v$, computes the update APSP matrix $A'$ of $G \setminus \{v\}$ in amortized time of $\tilde{O}(n^2)$. 
\end{thm}
Using this dynamic algorithm, one can support $O(n)$ vertex deletions, in total of $\tilde{O}(n^3)$ time. 
Throughout, we refer to the algorithm of Theorem \ref{thm:dyn-thorup} by Alg. $\DynamicAPSP$.
We focus on the efficient implementation of Step (1) of Alg. $\HopsetSmallHop$ by computing a nice path collection $\mathcal{Q}$. 
\begin{mdframed}[hidealllines=false,backgroundcolor=gray!00]
\center \textbf{Computation of the Nice Path Collection $\mathcal{Q}$}
\begin{flushleft}
\textbf{Input:} An $n$-vertex weighted $G=(V,E,\omega)$ with integer weights in $[1,W]$. \\
\textbf{Output:} A Nice Path collection $\mathcal{Q}$.
\end{flushleft}

\vspace{-2pt}
\begin{enumerate}
\item Apply the preprocessing algorithm of $\DynamicAPSP$ to compute $G^{\star}_{\omega}$. Let $B$ be the APSP matrix of $G$. 

\item Let $\widetilde{G}$ be the graph obtained from $G^{\star}_{\omega}$ by replacing each vertex $v_i \in G^{\star}_{\omega}$ by two copies $v^{\inn}_i$ and $v^{\outt}_i$. The two copies are connected by a directed edge $(v^{\inn}_i, v^{\outt}_i)$ of weight $1/n^2$. Connect the out-copy $v^{\outt}_j$ of every incoming neighbor $v_{j} \in N_{in}(v_i, G^{\star}_{\omega})$ to $v^{\inn}_i$. In the same manner, connect the in-copy $v^{\inn}_k$ of every outgoing neighbor $v_{k}\in N_{out}(v_i, G^{\star}_{\omega})$ to $v^{\outt}_i$. 

\item Apply the preprocessing algorithm of $\DynamicAPSP$ on $\widetilde{G}$, and let $A$ be the APSP matrix. Observe that $B_{i,j} = \left\lceil A_{v^{in}_i ,v^{out}_j } \right\rceil$. The weight assignment of $\widetilde{G}$ favors shortest paths of maximal number of hops.

\item Define a matrix $B'$ where $B'_{i,j}$ equals to the largest number of hops among all $v_{i}$-$v_j$ shortest paths in $\widetilde{G}$. The value $B'_{i,j}$ can be immediately deduced from the value of $A_{v^{in}_i ,v^{out}_j}$. 
That is, letting $A_{v^{\inn}_i ,v^{\outt}_j} = q-h/n^2$ for $q,h \in \mathbb{N}$, then $B'_{i,j} = h$.

\item Let $\mathcal{I}=\{ (i,j) ~\mid~ B'_{i,j} = \beta / 12\}$ and let $(i,j) \in \mathcal{I}$ be such that $A_{v^{in}_i ,v^{out}_j}$ is minimal overall pairs $(i',j')\in \mathcal{I}$. That is, 
$$A_{v^{\inn}_i ,v^{\outt}_j}\leq A_{v^{\inn}_{i'} ,v^{\outt}_{j'}}, \forall (i',j')\in \mathcal{I}~.$$
In the case where $|\mathcal{I}|=\emptyset$, abort. 

\item Define the corresponding $v_i$-$v_j$ shortest path $P \subseteq G$ in the following manner:
suppose that we already constructed the path up to vertex $v_{i'}$, then we seek in $B$ an entry $B_{i',j'}$ for which $B_{i',j'} + B_{j',j} = B_{i',j}$ and furthermore we require that 
$B'_{j',j} = B'_{i',j'} -1$. We choose the next vertex on the path to be $v_{j'}$.

\item Add the path $P$ to collection $\mathcal{Q}$.

\item Delete the vertices $v^{\inn}_{i}, v^{\outt}_i$ from $\widetilde{G}$ for every $v_i \in P$, by applying Alg. 
$\DynamicAPSP$. This results in an updated APSP matrix of $A$.

\item Delete the vertices of $P$ from graph $G^{\star}_{\omega}$ and delete the corresponding rows and columns from matrix $B$.
\item Goto step 4.
\end{enumerate}
\end{mdframed}
The paths in $\mathcal{Q}$ are ordered based on the insertion time. That is, denote $\mathcal{Q}=\{P'_1,\ldots, P'_k\}$ where $P'_i$ is the $i^{th}$ path added to $\mathcal{Q}$ in Step (7). 
\paragraph{Correctness.} We show that the output set $\mathcal{Q}$ is indeed nice, and satisfies properties (N1-N6).
\begin{itemize}
\item Property (N1) follows immediately.
\item Property (N2) follows from the fact that the path from $v_i$ to $v_j$ that we choose satisfies $B'_{i,j} = \beta/12$ and $B'_{i,j}$ is the number of hops in our path.
\item Property (N3) follows from Step (6) as in this step we construct a shortest path (by the fact that $B_{i,j}$ is the length of a shortest path from $v_i$ to $v_j$).
\item Properties (N4),(N5) follows from the fact that in each iteration of Step (5) we choose indices of a path of minimal length over the relevant choice.
That is we choose $(i,j)$ such that $A_{v^{in}_i ,v^{out}_j}$ is minimal overall pairs $(i',j')\in \mathcal{I}$ and the path length between $v_i$ and $v_j$ satisfies $B_{i,j} = \left\lceil A_{v^{in}_i ,v^{out}_j } \right\rceil$.
\item Property (N6) follows from the fact that when the procedure terminates the remaining graph $G^{\star}_{\omega}$ contains no shortest path on $\beta/12$ hops for otherwise the procedure would not have aborted in Step (5).
\end{itemize}

\paragraph{Running Time.}
We now go over all steps of the procedure above in order to bound the running time.
Step (1) runs in time $\tilde{O}(n^3)$ due to the amortized $\tilde{O}(n^2)$ running time of Alg. $\DynamicAPSP$ (see Theorem \ref{thm:dyn-thorup}). Step (2) is implemented in time $O(n^2)$.
By Theorem \ref{thm:dyn-thorup}, Step (3) is implemented in $\tilde{O}(n^3)$ time. 
Steps (4-6) are implemented in $O(n^2)$ time. Since these steps are implemented at most $n$ many times, it takes $O(n^3)$ in total. By Theorem \ref{thm:dyn-thorup} again, Step (7) runs in time  $\tilde{O}(n^3)$. 
Finally, Step (8) runs in $O(n^2)$ and since it is implemented for $O(n)$ times, we get $O(n^3)$ time in total. 
This completes the proof of Lemma \ref{lem:setQ}.

\subsection{Proof of Theorem \ref{thm:main_hopset_bound_large_D}}\label{app:missing-hopsets}
\APPENDLARGEHOP


\end{appendix}


\end{document}